\documentclass{fundam}

\publyear{2023}
\papernumber{1}
\volume{1}
\issue{1}

\usepackage{url} 
\usepackage[linesnumbered,ruled,lined]{algorithm2e}
\usepackage{graphicx}
\usepackage{tikz}
\usetikzlibrary{automata, positioning, arrows}

\usetikzlibrary{decorations.pathreplacing}

\begin{document}

\title{Compatibility graphs in scheduling on batch processing machines}

\author{Khaoula BOUAKAZ and Mourad BOUDHAR\\
	RECITS laboratory, Faculty of Mathematics, USTHB University, 
	 Algiers, Algeria\\
khaoulabouakaz19@gmail.com\\ mboudhar@yahoo.fr}

\address{RECITS laboratory, Faculty of Mathematics, USTHB University, BP 32 El-Alia, Bab Ezzouar, Algiers, Algeria}

\maketitle

\runninghead{K. BOUAKAZ and M. BOUDHAR}{Compatibility graphs in scheduling on batch processing machines} 

\begin{abstract} We consider the problem of minimizing the makespan on batch processing identical machines, subject to compatibility constraints, where two jobs are compatible if they can be processed simultaneously in a same batch. These constraints are modeled by  an undirected graph $G$, in which compatible jobs are represented by adjacent vertices. We show that several subproblems are polynomial. We propose some exact polynomial algorithms to solve these subproblems. To solve the general case, we propose a mixed-integer linear programming (MILP) formulation alongside with  heuristic approaches. Furthermore, computational experiments are carried out to measure the performance of the proposed methods.
\end{abstract}

\begin{keywords} batch scheduling, makespan, compatibility graph, complexity, matching
\end{keywords}

\section{Introduction}

In classical scheduling theory, it is assumed that machines can only process one job at a time.
However, in reality there is another scheduling model that can be called "batch machines" and
that refers to batches of jobs to be processed together (material painted together, material rolled together, material transported together, etc). In these cases, the jobs in the same batch have to be compatible.

The principal motivation for batch scheduling is the scheduling of burn-in operations in the
semiconductor industry where are exposed to high temperatures in a fixed capacity oven in order
to weed out chips susceptible to premature failure.

In this work, we consider the problem of scheduling a set $J=\lbrace J_1,...,J_n \rbrace$ of $n$ jobs  non-preemptively on batch processing identical machines to minimize the makespan $C_{max}$. We assume that the jobs are subject to compatibility constraints modeled by an undirected graph $G=(V, E)$, we call the compatibility graph, in which each job is represented by a vertex $(V=J)$, and each edge joins a pair of jobs that can be processed simultaneously in a same batch. By definition, a batch belongs to a clique of the graph $G$. The capacity $b$ of a batch and hence the clique size, may be finite $b = k$, variable, or infinite $b = \infty$ (it can process all jobs simultaneously). The job $J_i$ has the processing time $p_i$
and the processing time of a batch is equal to the maximum processing time or the sum of processing times of  jobs assigned to it. All the jobs in a batch must be available at the same date, start at the same date and finish at the same date. We assume that a setup-time $s$ must separate two successive batches and that the jobs of a batch are processed without setup-time. The setup-time is independent of the batches sequence, it is identical between each two successive batches on all machines.

A feasible schedule is an assignment of the jobs to the machines where  all the jobs in a same batch are adjacent in the graph  $G$. In a given schedule, the completion time of a batch $B_i$, denoted by $C_i$, is defined as the time by which the machine completes its processing. The largest one among them defines the makespan of the schedule $C_{max}$. We denote these batch scheduling problems with a fixed number of machines $m$ by $Bm,max| G = (V,E), b, p_i, s| C_{max}$ and $Bm,sum| G = (V,E), b, p_i,s|C_{max}$.

This paper is organized as follows. Section 2 is about the literature review.  In section 3,  we study the complexity results. Linear formulation  will be given  in Section 4. In Section 5, we present two heuristic approaches. Section 6 reports the computational results. We conclude this article in Section 7.

\section{Literature review}

The theory of batch scheduling for various scheduling objectives and additional constraints has been already well established in \cite{Reference6, Reference24}. Some very special cases of job compatibility, the Incompatible Families structures \cite{Reference10, Reference17, Reference19, Reference26} and the Compatible Families structures \cite{Reference6, Reference7, Reference14, Reference17} have previously been treated. Intensive research has subsequently been developed on this subject for various scheduling objectives and additional constraints, see for instance the surveys \cite{Reference6}. \cite{Reference20} provided a survey of scheduling research in semiconductor manufacturing. In \cite{Reference16}, the authors presented  the literature on parallel batching and  focused primarily on deterministic scheduling. They provided a taxonomy of parallel batching problems, distinguishing the compatible case  and the incompatible families.  They discussed scheduling approaches for single machine, parallel machines, and other environments such as flow shops and job shops.

For the case of single batch processing machine  with a general compatibility graph $G$ has been analyzed in \cite{Reference5}. In the same reference the authors have proved that the problem with an arbitrary compatibility graph $G$,  denoted $B1|G = (V,E), b = 2|C_{max}$ can be solved in polynomial time of order $O(n^3)$ as a maximum weighted matching problem. Note that the sub-problem with an identical processing time denoted $B1|G = (V,E), b = 2, p_i = 1|C_{max}$ can be solved in polynomial time of order $O(n^{2.5})$ as a maximum cardinality matching problem. The special cases with bipartite compatibility graph and split compatibility graph have been analyzed respectively in \cite{Reference3,Reference2,Reference100}. In \cite{Reference2}, each job has a release date and a processing time equal to 1. \cite{Reference13} studied the scheduling  on a single max-batch machine  with variable capacity, each job has a minimal processing time and the compatibility constraint may be represented by an interval graph. They considered several models with varying batch capacities, processing times or compatibility graphs. They summarized known results, and presented a min–max formula and polynomial time algorithms. 

In \cite{Reference15} is studied the case of  the  serial batching and  parallel batching. They did not introduce the notion of compatibility graphs. They developed  exact algorithms to minimize makespan on single and parallel batch processing machines for the problem  $1|s_j, B|C_{max}$ and $Pm|s_j, B|C_{max}$ where each job $j$  has a size $s_j$ and a minimum processing time $p_j$, which does not differ across the machines. Then, a subset of jobs that can be processed simultaneously forms a batch under the condition that their total size is no larger than machine capacity $S$. In \cite{Reference25} is considered the online (over time) scheduling problem of minimizing the makespan on $m$ unbounded parallel-batch machines with a compatibility constraint  which is represented by an interval compatibility graph. They  provided that there exists no online algorithm with a competitive ratio less than 2 and an online algorithm with a competitive ratio $2+(m-1)/(m+1)$, which is optimal for the case $m = 1$. When all jobs have the same processing times, they also gave an optimal online algorithm. Based on previous studies that often combine machine learning algorithms with practical production applications. \cite{Reference27} provided a new approach to the scheduling process using an advanced genetic algorithm.

For the case of ﬂowshop batching machines, \cite{Reference23} studied the no-wait flowshop problem with two batching machines, and proposed a polynomial algorithm for the above problem. They also extended their studies to the case of $m$ batching machines.  In \cite{Reference1} is considered a two-stage hybrid flowshop problem in which the first stage contains several identical discrete machines, and the second stage contains several identical batching machines. They studied the case of an interval graph when  which jobs have the same processing time on the first stage, for which a polynomial approximation scheme algorithm is presented. \cite{Reference22} studied  the problem of job scheduling in a flowshop with two  machines. The first machine is a discrete machine and the second machine is a batching machine  with the additional feature that the jobs of the same batch have to be compatible. A compatibility constraint  is defined by an  interval graph. They studied the complexity of the makespan minimization. They developed a heuristic approach and evaluated it. Other related results for the case of flowshop batching machines were presented by \cite{Reference18, Reference21}.

\section{Complexity results}

We, here, briefly present  new results regarding the  problems $B1, max|G=(V, E), b=2, s|C_{max}$ and  $ B1, sum|G=(V, E), b=2, s|C_{max}$ for the case of a single machine and $Bm, max|G=(V, E), b=2, p_i=p, s|C_{max}$, $Bm, sum|G=(V, E), b=2, p_i=p, s|C_{max}$ and $B2, max| G=(V, E), b=2, p_i=\lbrace p, q\rbrace, s| C_{max}$ for the case of several machines. These  problems  with an arbitrary compatibility graph $G$  can be solved in polynomial time.

\subsection{Case of  single machine}
We show in this section that  the problem of scheduling a set of $n$ jobs  $J_1, . . . , J_n$
non-preemptively on a single batch processing   machine  to minimize the makespan is polynomial. Here the jobs are subject to compatibility constraints modeled by an undirected graph $G$, each batch has a capacity  $b=2$, each job $J_i$ has a processing time $p_i$ and there exist a setup-time $s$ between each two successive  batches.

\begin{theorem}
The problem $B1, max|G=(V, E), b=2, s|C_{max}$ reduces to the maximum weighted matching.
\end{theorem}

\begin{proof}
Let $\sigma$ be a feasible schedule, it is composed of two-job batches and single-job batches. The two-job batches correspond to a matching $M$ in graph $G$ and the single-job batches are those of the set $V\setminus V_{M}$ ($V_{M}$ represents the  jobs of the matching $M$)  and contribute to $\sigma$ by a time $K=\sum\limits_{l/J_l \in  V\setminus V_{M}} p(J_l)$.

Let $c(e_{k})= \max\lbrace p(J_{i}),p(J_{j})\rbrace$ be the cost of the edge $e_k$, if the edge $e_k$ is incident to the vertices $J_i$
and $J_j$. The makespan
is written: $C_{max}(\sigma)$=$ \sum\limits_{k/e_{k}\in M} c(e_{k}) +K+s(\vert V\vert-\vert M\vert -1)$, where the expression $s(\vert V\vert-\vert M\vert -1)$ represents the global setup-time of $\sigma$.

We have
$$K=\sum_{l/ J_l \in V}p(J_l)-\sum_{l/J_{l} \in V_{M}}p(J_{l})
=\sum_{l/ J_l \in V}p(J_l)-\sum_{k/e_{k}=\left(  J_{i},J_{j}\right) \in M}\left(p(J_{i})+p(J_{j})\right)$$

Since $\vert M\vert=\sum_{k/e_{k}\in M}1$, then we obtain:

\begin{align*}
 C_{max}(\sigma) &=\sum_{k/e_{k}=( J_{i},J_{j})\in M} \max\lbrace p(J_{i}),p(J_{j})\rbrace +\sum_{l/J_l \in V}p(J_l) -\sum_{k/e_{k}=( J_{i},J_{j}) \in M}(p(J_{i})+p(J_{j}))\\
&\ \ \ +s\left(n-\sum_{k/e_{k}\in M}1-1 \right)\\
 &= \sum_{k/e_{k}=( J_{i},J_{j})\in M}\left( \max\lbrace p(J_{i}),p(J_{j})\rbrace-p(J_{i})-p(J_{j})-s\right) +\sum_{l/J_l \in V}p(J_l)+s(n-1)\\
 &= \sum_{k/e_{k}=( J_{i},J_{j})\in M}-\left( \min\lbrace p(J_{i}),p(J_{j})\rbrace+s\right) +\sum_{l/J_l \in V}p(J_l)+s(n-1).
\end{align*}
	
We consider a new weight function $\alpha$ on the graph $G$ such that for each edge $e_{k}=(J_{i},J_{j})$, $\alpha (e_{k})=  \min\lbrace p(J_{i}), p(J_{j})\rbrace+s$, then the makespan is written :

$$C_{max}(\sigma)= \sum_{l/J_l \in V}p(J_l)+s(n-1)-\sum_{k/e_{k}\in M}\alpha (e_{k})$$

$\sum\limits_{l/J_l \in V}p(J_l)+s(n-1)$ is a constant. Then, minimizing $C_{max}$ is equivalent to maximizing $\sum\limits_{k/e_{k}\in M}\alpha (e_{k})$.

Hence, the optimal solution of the problem is obtained by finding a  maximum weighted matching $M$ in the valued graph $H_{\alpha}=(G,\alpha)$.
\end{proof}

Algorithm \ref{alg1} solves the problem $B1, max|G=(V, E), b=2, s|C_{max}$.

\begin{algorithm}[h!]
\caption{}\label{alg1}
\KwIn{$G=(V, E)$, $p_i$, $s$} 
\KwResult{ schedule $\sigma$}
\begin{enumerate}
\item From the graph $G=(V, E)$, construct a new valued graph $H_{\alpha}=(G,\alpha)$ where each edge $e=(J_i, J_j)\in E$ is valued by $\alpha (e)=\min\lbrace p(J_{i}), p(J_{j}) \rbrace +s$.
\item  Find a maximum weighted matching $M$ in the graph  $H_{\alpha}$.
\item  Form the batches of $\sigma$:
\begin{itemize}
\item[•]For each edge of the matching $M$, process the corresponding two jobs in a same
batch.
\item[•]The other jobs are processed in single job batches. (Schedule
the batches in an arbitrary order).
\end{itemize}
\item    $C_{max}(\sigma)= \sum_{l/J_l \in V}p(J_l)+s(n-1)-\sum_{k/e_{k}\in M}\alpha (e_{k}).$
\end{enumerate}
\end{algorithm}

The best known algorithm for the maximum weighted matching is in $O(n^{3})$. Hence, also the Algorithm \ref{alg1} runs in $O(n^3)$.

If all the processing times are identical  to $p$, $\alpha(e)=p+s$ for all $e \in E$. The problem 
$B1, max|G=(V, E), b=2, pi=p, s|C_{max}$ is solved by finding a  maximum cardinality  matching  in the  graph $G$ with complexity $O(n^{2.5})$.

\begin{theorem}
The problem $B1, sum|G=(V, E), b=2, s|C_{max}$ reduces to the maximum cardinality matching.
\end{theorem}

\begin{proof}
Let $\sigma$ be a feasible schedule, it is composed of two-job batches and single-job batches. The two-job batches correspond to a matching $M$ in graph $G$ and the single-job batches are those of the set $V\setminus V_{M}$ and contribute to $\sigma$ by a time $K=\sum_{l/J_l \in  V\setminus V_{M}}p(J_l)$.

Let $c(e_{k})= p(J_{i})+p(J_{j})$ be the cost of the edge $e_k$, if the edge $e_k$ is incident to the vertices $J_i$
and $J_j$. The makespan
is written: $C_{max}(\sigma)$=$ \sum\limits_{k/e_{k}\in M} c(e_{k}) +K+s(\vert V\vert-\vert M\vert -1)$, where the expression $s(\vert V\vert-\vert M\vert -1)$ represents the global setup-time of $\sigma$.

We have
$$K=\sum_{l/J_l \in V}p(J_l)-\sum_{l/J_{l} \in V_{M}}p(J_{l})
=\sum_{l/J_l \in V}p(J_l)-\sum_{k/e_{k}=(  J_{i},J_{j}) \in M}(p(J_{i})+p(J_{j}))$$

Since $\vert M\vert=\sum_{k/e_{k}\in M}1$, then we obtain:

\begin{align*}
C_{max}(\sigma) & =\sum_{k/e_{k}=( J_{i},J_{j})\in M} ( p(J_{i})+p(J_{j})) +\sum_{l/J_l \in V}p(J_l) -\sum_{k/e_{k}=( J_{i},J_{j}) \in M}(p(J_{i})+p(J_{j}))\\
&\ \ \ +s\left(n-\sum_{k/e_{k}\in M}1 -1\right)\\
&= \sum_{k/e_{k}=( J_{i},J_{j})\in M}\left( p(J_{i})+p(J_{j})-p(J_{i})-p(J_{j})-s\right) +\sum_{l/J_l \in V}p(J_l)+s(n-1)\\
&= \sum_{l/J_l \in V}p(J_l)+s(n-1)-s\sum_{k/e_{k}=( J_{i},J_{j})\in M}1
\end{align*}
	
$\sum\limits_{l/J_l \in V}p(J_l)+s(n-1)$ is a constant. Then, minimizing $C_{max}$ is equivalent to maximizing $\sum\limits_{k/e_{k}\in M}1$.
Hence, the optimal solution of the problem is obtained by finding a  maximum cardinality matching $M$ in the  graph $G$.
\end{proof}

Algorithm \ref{alg2} solves the problem $B1, sum|G=(V, E), b=2, s|C_{max}$.

\begin{algorithm}[h!]
\caption{}\label{alg2}
\KwIn{$G=(V, E)$, $p_i$, $s$} 
\KwResult{ schedule $\sigma$}
\begin{enumerate}
\item  Find a maximum cardinality  matching $M$ in the graph  $G$.
\item  Form the batches of $\sigma$:
\begin{itemize}
\item[•]For each edge of the matching $M$, process the corresponding two jobs in a same
batch.
\item[•]The other jobs are processed in single job batches. (Schedule the batches in an arbitrary order).
\end{itemize}
\item  $C_{max}(\sigma)= \sum_{l/J_l \in V}p(J_l)+s(n-\vert M \vert-1)$.
\end{enumerate}
\end{algorithm}

The best known algorithm for the maximum cardinality matching is in $O(n^{2.5})$. Hence, also the Algorithm \ref{alg2} runs in $O(n^{2.5})$.

\subsection{Case of several machines }

We show in this section that  the scheduling problems $Bm, max| G=(V, E), b=2, p_i=p, s|C_{max}$,  $Bm, sum| G=(V, E), b=2, p_i=p, s|C_{max}$ and  $B2, max| G=(V, E), b=2, p_i=\lbrace p,q\rbrace, s|C_{max}$ are  solvable in polynomial time. Where the jobs are subject to compatibility constraints modeled by an undirected graph $G$, each batch has a capacity  $b=2$,  all the processing times are identical  to $p$  or to two values $\lbrace p, q \rbrace$  where $p<q$  and there exist  a setup-time $s$ between each two successive batches. The problem  $Bm| G = (V, E), b=1| C_{max}$ is equivalent to the parallel machine problem
$Pm|| C_{max}$. Consequently, these batch scheduling problems are NP-hard for $m\geq 2$.

\begin{theorem}
The problem $Bm, max|G=(V, E), b=2, p_i=p, s|C_{max}$ reduces to the maximum cardinality matching.
\end{theorem}

\begin{proof}
The value of an optimal solution of the problem $B1, max|G=(V, E), b=2, p_i=p, s|C_{max}$ devised by $m$ is a lower bound of the problem $Bm, max|G=(V, E), b=2, p_i=p, s|C_{max}$ that is $LB=\left\lceil \frac{n-\vert M\vert}{m}\right\rceil p+\left(\left\lceil \frac{n-\vert M\vert}{m}\right\rceil-1\right)s$.  Since all the processing times are equal to $p$, then to build  an optimal schedule for the above problem, it suffices to find a maximum matching $M$ in the graph $G$, the value of the corresponding solution of $M$ is equal to $C_{max}=\left\lceil \frac{n-\vert M\vert}{m}\right\rceil p+\left(\left\lceil \frac{n-\vert M\vert}{m}\right\rceil-1\right)s$ which is equal to $LB$.
\end{proof}
Algorithm \ref{alg3} solves the problem $Bm, max|G=(V, E), b=2, p_i=p, s|C_{max}$.

\begin{algorithm}[h!]
\caption{}\label{alg3}

\KwIn{$G=(V, E), p, s$} 
\KwResult{ schedule $\sigma$}
\begin{enumerate}
\item   Find a maximum cardinality matching $M$ in the graph $G$.
\item  Form the following batches of $\sigma$:
\begin{itemize}
\item[•]For each edge of the matching $M$, process the corresponding two jobs in a same
batch.
\item[•]The other jobs are processed in single job batches.
\end{itemize}
\item  Schedule the batches at the first available machine.
\item $C_{max}(\sigma)=\left\lceil \frac{n-\vert M\vert}{m}\right\rceil p+\left(\left\lceil \frac{n-\vert M\vert}{m}\right\rceil-1\right)s$.
\end{enumerate}
\end{algorithm}

The best known algorithm for the  maximum cardinality matching is in $O(n^{2.5})$. Hence, also,  the Algorithm \ref{alg3} runs in $O(n^{2.5})$.

\begin{theorem}
The problem $Bm, sum|G=(V, E), b=2, p_i=p, s|C_{max}$ reduces to the maximum cardinality matching.
\end{theorem}

\begin{proof}
Since all the processing times are equal to $p$, then to build  an optimal schedule for the above problem, it suffices to find a maximum matching $M$ in the graph $G$, to minimize the number of batches for the above problem  and schedule the two-job batches at the first one because its processing time is equal to $2p$ and then the  batches of single job, at the first available machine.

\begin{figure} [h!]
\begin{center}
\begin{tikzpicture}

\color{black}

\draw [draw=black] (0,0) rectangle (1,1);
\draw [draw=black] (0,1) rectangle (1,4);
\draw [draw=black] (0,4) rectangle (1,5);
\draw [draw=black] (1,0) rectangle (1.5,5);
\draw [draw=black] (1.5,0) rectangle (3.5,1);
\draw [draw=black] (1.5,1) rectangle (3.5,4);
\draw [draw=black] (1.5,4) rectangle (3.5,5);
\draw [draw=black] (3.5,0) rectangle (4,5);
\draw [draw=black] (4,0) rectangle (4.5,1);
\draw [draw=black] (4,1) rectangle (4.5,2);
\draw [draw=black] (4.5,0) rectangle (5,2);

\draw [draw=black] (5,0) rectangle (5.5,1);
\draw [draw=black] (5,1) rectangle (5.5,2);
\draw [draw=black] (4,2) rectangle (5,3);
\draw [draw=black] (4,3) rectangle (5,4);
\draw [draw=black] (4,4) rectangle (5,5);
\draw [draw=black] (5.5,0) rectangle (6,2);
\draw [draw=black] (5,2) rectangle (5.5,5);
\draw [draw=black] (6,0) rectangle (6.5,1);
\draw [draw=black] (6,1) rectangle (6.5,2);
\draw [draw=black] (5.5,2) rectangle (6,3);
\draw [draw=black] (5.5,3) rectangle (6,4);
\draw [draw=black] (5.5,4) rectangle (6,5);
\draw [draw=black] (6,2) rectangle (6.5,5);
\draw [draw=black] (6.5,0) rectangle (7,2);
\draw [draw=black] (7,0) rectangle (9,2);
\draw [draw=black] (6.5,2) rectangle (8.5,5);
\draw [draw=black] (8.5,2) rectangle (9,5);
\draw [draw=black] (9,0) rectangle (9.5,2);
\draw [draw=black] (9.5,0) rectangle (10,1);
\draw [draw=black] (9.5,1) rectangle (10,2);
\draw [draw=black] (9,2) rectangle (9.5,3);
\draw [draw=black] (9,3) rectangle (9.5,4);
\draw [draw=black] (9,4) rectangle (9.5,5);
\draw [draw=black] (6.5,4) rectangle (8.5,5);
\draw [draw=black] (6.5,3) rectangle (8.5,4);
\draw [draw=black] (,1) rectangle (9,1);

\node[scale=0.8] at (0.5,4.5) {$(J^p, J^p)$};
\node[scale=0.8] at (0.5,0.5) {$(J^p, J^p)$};
\node[] at (0.5,1.5) {$.$};
\node[] at (0.5,2) {$.$};
\node[] at (0.5,2.5) {$.$};
\node[] at (0.5,3) {$.$};
\node[] at (0.5,3.5) {$.$};
\node[] at (-0.5,4.5) {$M_1$};
\node[] at (-0.5,0.5) {$M_m$};
\node[] at (-0.5,1.5) {$.$};
\node[] at (-0.5,2.5) {$.$};
\node[] at (-0.5,2) {$.$};
\node[] at (-0.5,3.5) {$.$};
\node[] at (-0.5,3) {$.$};

\node[] at (2,0.5) {$.$};
\node[] at (2.5,0.5) {$.$};
\node[] at (3,0.5) {$.$};
\node[] at (2,4.5) {$.$};
\node[] at (2.5,4.5) {$.$};
\node[] at (3,4.5) {$.$};
\node[scale=0.8] at (4.5,4.5) {$(J^p, J^p)$};
\node[scale=0.8] at (4.5,2.5) {$(J^p, J^p)$};
\node[] at (4.5,3.25) {$.$};
\node[] at (4.5,3.5) {$.$};
\node[] at (4.5,3.75) {$.$};
\node[scale=0.8] at (4.25,0.5) {$(J^p)$};
\node[scale=0.8] at (4.25,1.5) {$(J^p)$};
\node[scale=0.8] at (5.25,0.5) {$(J^p)$};
\node[scale=0.8] at (5.25,1.5) {$(J^p)$};
\node[scale=0.8] at (6.25,0.5) {$(J^p)$};
\node[scale=0.8] at (6.25,1.5) {$(J^p)$};
\node[scale=0.8] at (5.75,2.5) {$(J^p)$};
\node[scale=0.8] at (5.75,4.5) {$(J^p)$};
\node[scale=0.8] at (9.75,0.5) {$(J^p)$};
\node[scale=0.8] at (9.75,1.5) {$(J^p)$};
\node[scale=0.8] at (9.75,0.5) {$(J^p)$};
\node[scale=0.8] at (9.25,2.5) {$(J^p)$};
\node[scale=0.8] at (9.25,4.5) {$(J^p)$};
\node[] at (7.5,0.5) {$.$};
\node[] at (8.5,0.5) {$.$};
\node[] at (8,0.5) {$.$};

\node[] at (7.5,1.5) {$.$};
\node[] at (8.5,1.5) {$.$};
\node[] at (8,1.5) {$.$};
\node[] at (7.5,2.5) {$.$};
\node[] at (8,2.5) {$.$};
\node[] at (7,2.5) {$.$};

\node[] at (7.5,3.5) {$.$};
\node[] at (8,3.5) {$.$};
\node[] at (7,3.5) {$.$};
\node[] at (7.5,4.5) {$.$};
\node[] at (8,4.5) {$.$};
\node[] at (7,4.5) {$.$};
\node[] at (7,2.25) {$.$};
\node[] at (7,2.75) {$.$};
\node[] at (7.5,2.75) {$.$};
\node[] at (7.5,2.25) {$.$};
\node[] at (8,2.75) {$.$};
\node[] at (8,2.25) {$.$};
\node[] at (7.5,3.75) {$.$};
\node[] at (7.5,3.25) {$.$};
\node[] at (8,3.75) {$.$};
\node[] at (8,3.25) {$.$};
\node[] at (7,3.75) {$.$};
\node[] at (7,3.25) {$.$};
\node[] at (7,4.75) {$.$};
\node[] at (7,4.25) {$.$};
\node[] at (7.5,4.75) {$.$};
\node[] at (7.5,4.25) {$.$};
\node[] at (8,4.75) {$.$};
\node[] at (8,4.25) {$.$};
\node[] at (7.5,1.75) {$.$};
\node[] at (7.5,1.25) {$.$};
\node[] at (8,1.75) {$.$};
\node[] at (8,1.25) {$.$};
\node[] at (8.5,1.75) {$.$};
\node[] at (8.5,1.25) {$.$};
\node[] at (7.5,0.75) {$.$};
\node[] at (7.5,0.25) {$.$};
\node[] at (8.5,0.75) {$.$};
\node[] at (8.5,0.25) {$.$};
\node[] at (8,0.75) {$.$};
\node[] at (8,0.25) {$.$};

\node[] at (2,3.5) {$.$};
\node[] at (2,3) {$.$};
\node[] at (2,2.5) {$.$};
\node[] at (2,2) {$.$};
\node[] at (2,1.5) {$.$};
\node[] at (2.5,3.5) {$.$};
\node[] at (2.5,3) {$.$};
\node[] at (2.5,2.5) {$.$};
\node[] at (2.5,2) {$.$};
\node[] at (2.5,1.5) {$.$};
\node[] at (3,3.5) {$.$};
\node[] at (3,3) {$.$};
\node[] at (3,2.5) {$.$};
\node[] at (3,2) {$.$};
\node[] at (3,1.5) {$.$};
\node[] at (5.75,3.25) {$.$};
\node[] at (5.75,3.5) {$.$};
\node[] at (5.75,3.75) {$.$};

\node[] at (9.25,3.25) {$.$};
\node[] at (9.25,3.5) {$.$};
\node[] at (9.25,3.75) {$.$};

\draw[color=black,decorate,decoration={brace}]
(0,5.25) -- (5,5.25) node[above=0.2cm,pos=0.5] {$2\left\lceil\frac{\vert M\vert}{m}\right\rceil p+ \left(\left\lceil\frac{\vert M\vert}{m}\right\rceil-1\right)s$};

\draw[color=black,decorate,decoration={brace,raise=0.1cm}]
(9.75,5) -- (9.75,2) node[above=0.2cm,pos=0.5,sloped] {$mod(\vert M\vert, m)$};

\draw[color=black,decorate,decoration={brace,raise=0.1cm}]
(10.25,2) -- (10.25,0) node[above=0.2cm,pos=0.5,sloped] {$\frac{k}{2}$};

\draw[color=black,decorate,decoration={brace}]
(5.5,5.25) -- (9.5,5.25) node[above=0.2cm,pos=0.5] {$\left\lceil\frac{n'}{m}\right\rceil p+ \left(\left\lceil\frac{n'}{m}\right\rceil-1\right)s$};

\draw (1,0) edge (1.5,1);
\draw (1.25,0) edge (1.5,0.5);
\draw (1,0.5) edge (1.5,1.5);
\draw (1,1) edge (1.5,2);
\draw (1,1.5) edge (1.5,2.5);
\draw (1,2.5) edge (1.5,3.5);
\draw (1,3.5) edge (1.5,4.5);
\draw (1,4.5) edge (1.25,5);
\draw (1,2) edge (1.5,3);
\draw (1,3) edge (1.5,4);
\draw (1,4) edge (1.5,5);

\draw (3.5,0) edge (4,1);
\draw (3.75,0) edge (4,0.5);
\draw (3.5,0.5) edge (4,1.5);
\draw (3.5,1.5) edge (4,2.5);
\draw (3.5,2.5) edge (4,3.5);
\draw (3.5,3.5) edge (4,4.5);
\draw (3.5,4.5) edge (3.75,5);
\draw (3.5,1) edge (4,2);
\draw (3.5,2) edge (4,3);
\draw (3.5,3) edge (4,4);
\draw (3.5,4) edge (4,5);

\draw (4.5,0) edge (5,1);
\draw (4.75,0) edge (5,0.5);
\draw (4.5,0.5) edge (5,1.5);
\draw (4.5,1.5) edge (4.75,2);
\draw (4.5,1) edge (5,2);

\draw (9,0) edge (9.5,1);
\draw (9.25,0) edge (9.5,0.5);
\draw (9,0.5) edge (9.5,1.5);
\draw (9,1) edge (9.5,2);
\draw (9,1.5) edge (9.25,2);

\draw (6.5,0) edge (7,1);
\draw (6.75,0) edge (7,0.5);
\draw (6.5,0.5) edge (7,1.5);
\draw (6.5,1) edge (7,2);
\draw (6.5,1.5) edge (6.75,2);

\draw (5.5,0) edge (6,1);
\draw (5.75,0) edge (6,0.5);
\draw (5.5,0.5) edge (6,1.5);
\draw (5.5,1) edge (6,2);
\draw (5.5,1.5) edge (5.75,2);

\draw (6,2) edge (6.5,3);
\draw (6.25,2) edge (6.5,2.5);
\draw (6,2.5) edge (6.5,3.5);
\draw (6,3) edge (6.5,4);
\draw (6,3.5) edge (6.5,4.5);
\draw (6,4) edge (6.5,5);
\draw (6,4.5) edge (6.25,5);

\draw (5,2) edge (5.5,3);
\draw (5.25,2) edge (5.5,2.5);
\draw (5,2.5) edge (5.5,3.5);
\draw (5,3) edge (5.5,4);
\draw (5,3.5) edge (5.5,4.5);
\draw (5,4) edge (5.5,5);
\draw (5,4.5) edge (5.25,5);

\draw (8.5,2) edge (9,3);
\draw (8.75,2) edge (9,2.5);
\draw (8.5,2.5) edge (9,3.5);
\draw (8.5,3) edge (9,4);
\draw (8.5,3.5) edge (9,4.5);
\draw (8.5,4) edge (9,5);
\draw (8.5,4.5) edge (8.75,5);

 \end{tikzpicture}
    \caption{Structure of the schedule ($(J^p, J^p)$ is a two-job batch and $(J^p)$ a single job batch).}\label{fig000}
\end{center}
\end{figure}

The makespan of the two-job batches is equal to $2\left\lceil \frac{\vert M\vert}{m} \right\rceil p+\left(\left\lceil \frac{\vert M\vert}{m} \right\rceil-1\right)s$. 

Let $k= 2\left(\left(\left\lceil\frac{\vert M\vert}{m}\right\rceil-\left\lfloor\frac{\vert M\vert}{m}\right\rfloor\right)m-mod(\vert M\vert, m)\right)$, be the number of the free places (of length $p$) in the interval $\left[2\left\lfloor \frac{\vert M\vert}{m} \right\rfloor p,2\left\lceil \frac{\vert M\vert}{m}\right\rceil p\right]$. $\left(k=\left\{\begin{tabular}{ll}
	$0$  &  if $ mod(\vert M\vert, m)=0$\\
	$2(m-mod(\vert M\vert, m))$  & otherwise \\
\end{tabular} \right.\right)$

If the number of the remaining jobs ($n-2\vert M \vert$ jobs) is less than or equal to $k/ 2$ so the makespan is equal to $2\left\lceil \frac{\vert M\vert}{m}\right\rceil p+\left(\left\lceil \frac{\vert M\vert}{m}\right\rceil-1\right)s$, but if it's between $k/2$ and $k$ so the batch number of the previous solution will increase by one.

If the  cardinal matching of $M$ is multiple of $m$  or the number of remaining jobs ($n-2\vert M \vert$ jobs) is greater than the free places $k$ and $mod \footnote{The function $mod(x, y)$ is the division remainder of  $x$ by $y$}(n',m)$ less than or equal to $mod(\vert M \vert, m)$ where  $n'=n-2\vert M\vert-k$, so  the makespan is equal to $\left(2\left\lceil \frac{\vert M\vert}{m} \right\rceil+\left\lceil  \frac{n'}{m}  \right\rceil\right) p+ \left(\left\lceil \frac{\vert M\vert}{m}\right\rceil+\left\lceil \frac{n'}{m}  \right\rceil-1\right)s$, but if it's  greater than $mod(\vert M \vert, m)$ or $mod(n',m)=0$, the batch number of the previous solution will increase by one. (See Figure \ref{fig000}).
\end{proof}

Algorithm \ref{alg4} solves the problem $Bm, sum|G=(V, E), b=2, p_i=p, s|C_{max}$.

\begin{algorithm}[h!]
\caption{}\label{alg4}

\KwIn{$G=(V, E), p, s$} 
\KwResult{ schedule $\sigma$}
\begin{enumerate}
\item   Find a maximum cardinality matching $M$ in the graph $G$.
\item  Form the following batches of $\sigma$:
\begin{itemize}
\item[•]For each edge of the matching $M$, process the corresponding two jobs in a same
batch.
\item[•]The other jobs are processed in single job batches.
\end{itemize}
\item  Schedule the batches at the first available machine starting by the two-job batches.
\item The makespan equal to: \\

$C_{max}=    2\left\lceil\frac{\vert M\vert}{m}\right\rceil p+ \left(\left\lceil\frac{\vert M\vert}{m}\right\rceil-1\right)s+$\\
$\left\{\begin{tabular}{lll}
$s$ &  if & $\dfrac{k}{2}<n-2\vert M \vert  \leq k$\\
$s+ \left\lceil\frac{n'}{m}\right\rceil p+ \left(\left\lceil\frac{n'}{m}\right\rceil-1\right)s$ & if &  $ n'>0 $ and $ \left(mod(n',m)\leq mod(\vert M \vert, m) \mbox{ or }\right.$ \\ & & $\left.mod(\vert M \vert, m)=0\right)$\\
$2s+ \left\lceil\frac{n'}{m}\right\rceil p+ \left(\left\lceil\frac{n'}{m}\right\rceil-1\right)s$ & if & $ n'>0 $ and $ \left(mod(n',m)> mod(\vert M \vert, m) \mbox{ or }\right.$\\ & & $\left.mod(n',m)=0\right)$
\end{tabular} \right. 
$
\end{enumerate}
\end{algorithm}

The best known algorithm for the  maximum cardinality matching is in $O(n^{2.5})$. Hence, also,  the Algorithm \ref{alg4} runs in $O(n^{2.5})$.

Before stating the next theorem, we start by describing the notations used and the Algorithm \ref{alg5}.
\begin{itemize}
\item[•] $B_i$: the batch $i$.
\item[•] $J_i^p$: the job of processing time $p$ of  batch $i$.
\item[•] $J_i^q$: the job of processing time $q$ of  batch $i$.
\item[•] $\sigma$: the solution given  by Algorithm \ref{alg5}.
\item[•] $Z(\sigma)$: is the makespan of $\sigma$.
\item[•] $n_p$: the number of batches of processing time $p$ of $\sigma$.
\item[•] $n_q$: the number of batches of processing time $q$ of $\sigma$.
\item[•] $n_{pq}$: the number of two-job batches, one of processing time  $p$ and the other of processing time $q$ of $\sigma$.
\end{itemize}

Algorithm \ref{alg5} described bellow is the equivalent of Algorithm \ref{alg1} in the case of two machines.

\begin{algorithm}[h!]
\caption{IS}\label{alg5}

\KwIn{$G=(V, E), p, q, s$} 
\KwResult{ schedule $\sigma$}
\begin{enumerate}
\item   From  graph $G = (V, E) $, construct a new valued graph $H_{\alpha} = (G,\alpha)$ where each edge $e = (J_i, J_j) \in E$ is valued by $\alpha(e) = min\lbrace p(J_i), p(J_j)\rbrace+s$.
\item  Find a maximum weighted matching $M$ in the graph $H_{\alpha}$.
\item  Form the following batches of $\sigma$:
\begin{itemize}
\item[•]For each edge of the matching $M$, process the corresponding two jobs in a same
batch.
\item[•]The other jobs are processed in single job batches.
\end{itemize}
\item  Schedule the batches according to $P2|p_i\in\{p,q\}|C_{max}$ with only two processing times.
\end{enumerate}
\end{algorithm}

\begin{theorem} 
	The problem $B2,max|G = (V, E), b = 2, p_i \in \{p, q\}, s|C_{max}$ is solvable in polynomial time by Algorithm \ref{Tux} in $O(n^{9})$ time.
\end{theorem}

\begin{algorithm}[h!]
\caption{}
\label{Tux}
\KwIn{$G=(V,E)$,  $p$, $q$, $s$} 
\KwResult{ $\sigma^{*}$}

Apply the Algorithm \ref{alg5}, let $\sigma$  be the given solution\;
$\sigma^{*}:=\sigma$\;

  \If{ (  $(n_{pq}\geq 2)$ and ($n_p+n_q$ is odd) and  $(2p+s\geq q)$  ) }
  { \For{( $J_i^q \in V$ and $J_j^q \in V$)}
  {  \For{( $J_k^p \in V$ and $J_l^p \in V$ such that $(J_k^p, J_i^q)\in E$ and $(J_l^p, J_j^q)\in E$)}
  {
  $E^{'}:=E\setminus\lbrace (J_k^p, J^q)$ and $ (J_l^p, J^q), J^q \in V  \rbrace$\;
  \For{$J_e^{q} \in V$ and $J_f^{q} \in V$}
  {\If{($(J_i^q, J_e^{q})\in E$ and $(J_j^{q}, J_f^{q}) \in E$)}
  { $V^{'}:=V\setminus\lbrace J_i^q, J_j^q, J_e^{q}, J_f^{q}  \rbrace$\;
    Apply  Algorithm \ref{alg5} ($G=(V^{'}; E^{'}), p, q, s$)\;
    Let $\sigma^{'}$ be the solution given    by the Algorithm \ref{alg5}, add to $\sigma^{'}$ the two batches $(J_i^q, J_e^{q})$ and $(J_j^{q}, J_f^{q})$ at the time 0 and add  $q$ to $ C_{max}(\sigma^{'})$ \;
    \If{$C_{max}(\sigma^{'})<C_{max}(\sigma^{*})$}
    {$\sigma^{*}:=\sigma^{'}$\;
    }
  }
  }
  }
  }
  \Else{
   $\sigma$ is an optimal solution\;
  } 
  }
\end{algorithm}

\begin{proof}
It is clear that the only possible case to decrease the value of $C_{max}(\sigma)$ is to increase the number of batches with  processing time $p$ and to decrease the number of batches with  processing time $q$ for  $\sigma$. To realize the decrease, we  process two jobs with processing time  $q$ of two different batches in a same batch. If one of the two batches is a singleton and the other is of type $(J^p, J^q)$, then absurd with $\sigma$ contains the  maximum weighted matching, so to decrease $n_q$, we use at least two batches $B_i=(J_i^p, J_i^q)$ and $B_j=(J_j^p, J_j^q)$ of $\sigma$ (See Figure \ref{fig00}).

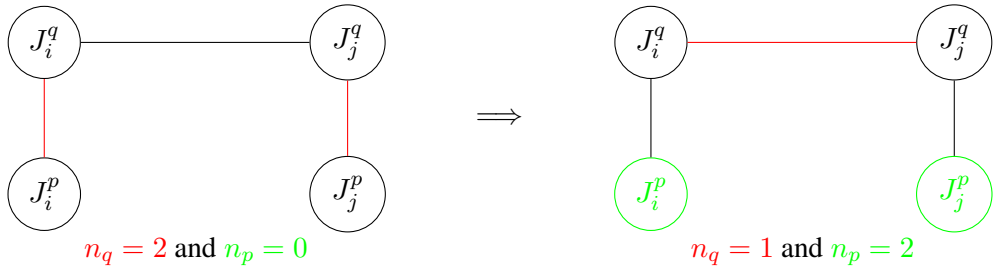
\begin{figure}[h!]
\begin{center}
\begin{tikzpicture}
\color{black}

\node[draw,circle] (T1) at (-4,0) {$J_i^q$};

\node[color=black][draw,circle] (T2) at (0,0) {$\color{black}J_j^q$};
\node[color=black][draw,circle] (T3) at (4,0) {$J_i^q$};
\node[color=black][draw,circle] (T4) at (8,0) {$J_j^q$};
\node[color=black][draw,circle] (T5) at (-4,-2) {$J_i^p$};
\node[color=black][draw,circle] (T6) at (0,-2) {$J_j^p$};
\node [color=green] [draw,circle] (T7) at (4,-2) {$J_i^p$};
\node [color=green] [draw,circle] (T8) at (8,-2) {$J_j^p$};

\draw [color=black](T1) edge (T2);
\draw [color=red] (T1) edge (T5); 
\draw [color=red] (T2) edge (T6);
\draw [color=red](T3) edge (T4); 
\draw [color=black] (T3) edge (T7);
\draw [color=black] (T4) edge (T8);
 

\node[] at (2,-1) {$\Longrightarrow$};
\node[] at (-2,-2.75) {$\color{red}n_q=2$ and $\color{green}n_p=0$};
\node[] at (6,-2.75) {$\color{red}n_q=1$ and $\color{green}n_p=2$};
 
 \end{tikzpicture}
    \caption{Structure of the batches.}\label{fig00}
\end{center}
\end{figure}

\vspace*{-0.5cm}
\begin{enumerate}
\item \textbf{Case 1:} \textbf{$n_q$ even and $n_p$ even}

A lower bound on this problem is given by the value of an optimal solution of the problem on a single machine minus an $s$ and  divided by two, $LB=\dfrac{n_{q}q+n_{p}p+(n_q+n_{p}-1)s-s} {2}$. If $n_q$ and $n_{p}$ are even, then $C_{max}(\sigma)=\dfrac{n_q}{2}q+\dfrac{n_{p}}{2}p+\left(\dfrac{n_q}{2}+\dfrac{n_{p}}{2}-1\right)s=LB$. Therefore, $\sigma$ is an optimal solution.

\item \textbf{Case 2: $n_q$ odd and $n_p$ odd}

\begin{itemize}
\item \textbf{Case 2.1:}
if $2p+s<q$, then we cannot find two batches $B_i=(J_i^{p}, J_i^{q})$ and $B_j=(J_j^{p}, J_j^{q})$ with $(J_i^{q}, J_j^{q})\in E$, otherwise we will have a contradiction with $\sigma$ constructed by Algorithm \ref{alg5}, corresponding to the maximum weighted matching. So $\sigma$ is an optimal solution. (See Figure \ref{fig0}).

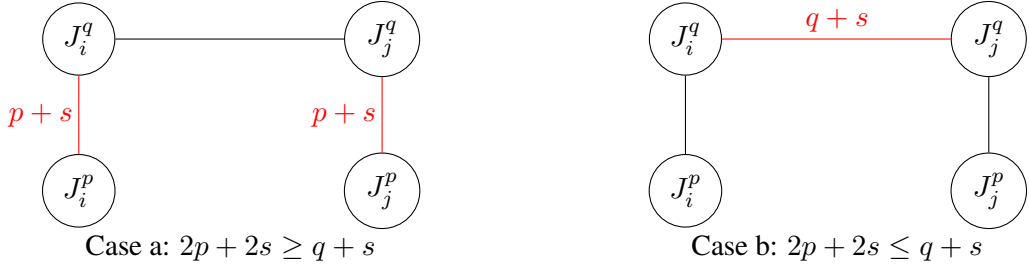
\begin{figure} [h!]
\begin{center}
\begin{tikzpicture}
\color{black}
\node[draw,circle] (T1) at (-4,0) {$J_i^q$};

\node[draw,circle] (T2) at (0,0) {\color{black}$J_j^q$};
\node[draw,circle] (T3) at (4,0) {$J_i^q$};
\node[draw,circle] (T4) at (8,0) {$J_j^q$};
\node[draw,circle] (T5) at (-4,-2) {$J_i^p$};
\node[draw,circle] (T6) at (0,-2) {$J_j^p$};
\node [] [draw,circle] (T7) at (4,-2) {$J_i^p$};

\node [] [draw,circle] (T8) at (8,-2) {$J_j^p$};

\draw (T1) edge (T2);
\draw [color=red] (T1) edge (T5);  
\draw [color=red] (T2) edge (T6);
\draw [color=red](T3) edge (T4); 
\draw (T3) edge (T7);
\draw (T4) edge (T8);

\node[] at (-4.5,-1) {$\color{red}p+s$};
\node[] at (-0.5,-1) {$\color{red}p+s$};
\node[] at (6,0.25) {$\color{red}q+s$};
\node[] at (-2,-2.75) {Case a: $2p+2s\geq q+s$};
\node[] at (6,-2.75) {Case b: $2p+2s\leq q+s$};

 \end{tikzpicture}
    \caption{Structure of the maximum weighted matching.}\label{fig0}
  \end{center}
\end{figure}

\item \textbf{Case 2.2:} if $2p+s\geq q$,  we suppose that we  deleted from $\sigma$, $a_1$ batches of  processing time $q$ and we  added $(a_1+1)$ batches of   processing time $p$.  Then the matching value corresponds to these jobs in $\sigma$ is equal to $(2 a_1 p+2 a_1 s)$, and its matching value in the new solution $\sigma_1$ is equal to $a_1 q+(a_1-1) p+(2 a_1-1)s$. We have by hypothesis $2a_1 p+2a_1 s\geq a_1 q+(a_1-1) p+(2 a_1-1)s$, so $(a_1+1) p+ s\geq a_1 q$, then:

 \begin{itemize}
 \item[•] \textbf{If $a_1$ is even:} suppose $n_q=2K+1$ and $n_p=2C+1$, then  $Z(\sigma)=(K+1)q+C p+(K+C)s$ and $Z(\sigma_1)=\left(\frac{2K-a_1}{2}\right)q+\left(C+\frac{a_1}{2}+2\right)p+(K+C+1)s$. Thus,  $Z(\sigma)\leq Z(\sigma_1)$ if, and only if, $\left(\frac{a_1}{2}+1\right)q \leq   \left(\frac{a_1}{2}+2\right)p+s$. We have by hypothesis $a_1 q\leq (a_1+1) p+ s$, then $a_1 q-\left(\frac{a_1}{2}-1\right)q < (a_1+1) p+ s-(\frac{a_1}{2}-1)p$ because $p<q$, so $Z(\sigma)< Z(\sigma_1)$.

 \item[•]  \textbf{If $a_1$ is odd:} then $Z(\sigma)=(K+1)q+Cp+(K+C)s$ and $Z(\sigma_1)=\left(\frac{2K-a_1+1}{2}\right)q+\left(\frac{2C+a_1+3}{2}\right)p+(K+C+1)s$. Thus, $Z(\sigma)\leq Z(\sigma_1)$ if, and only if, $\left(\frac{a_1+1}{2}\right)q \leq \left(\frac{a_1+1}{2}+1\right)p+s $. We have by hypothesis $a_1 q\leq (a_1+1) p+ s$, then $a_1 q -\left(\frac{a_1-1}{2}\right)q < (a_1+1) p+ s-\left(\frac{a_1-1}{2}\right)p$, so $Z(\sigma)<Z(\sigma_1)$. 
 \end{itemize}
  
 Therefore $\sigma$ is an optimal solution.
\end{itemize}

\item \textbf{Case 3: $n_q$ even and $n_p$ odd or $n_q$ odd and $n_p$ even}
\begin{itemize}
\item \textbf{Case 3.1:} if $2p+s<q$, similar to case 2.1.
\item \textbf{Case 3.2:} if $2p+s\geq q$, see Lemma \ref{lem1} (here $n_p+n_q$ is odd).
\end{itemize}
\end{enumerate}

\begin{lemma}\label{lem1}
If $n_p+n_q$ is odd and if we delete, from $\sigma$, $a_1$ batches of processing time $q$ and we add $(a_1+1)$ batches of processing time $p$, we obtain a new solution $\sigma_1$, such that $Z(\sigma)> Z(\sigma_1)$.
\end{lemma}

\begin{proof} 
Let $n'_q=n_q-a_1$ and $n'_p=n_{p}+a_1 +1$ be the new numbers of batches. There are four possible types of scheduling depending on the parity of $n'_p$ and $n'_q$, see Figure \ref{ft}.

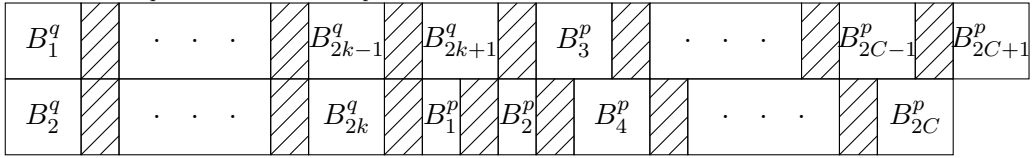
\begin{figure}[h!]
\begin{center}
\begin{tikzpicture}[scale=1]
\color{black}

\node[right] at (0,2.25) {- Type 1: $n'_q=2K$ and $n'_p=2C$ with $Z(\sigma)=Kq+Cp+(K+C-1)s$};

\draw [draw=black] (0,0) rectangle (1,1);
\draw [draw=black] (0,1) rectangle (1,2);
\draw [draw=black] (1,0) rectangle (1.5,1);
\draw [draw=black] (1,1) rectangle (1.5,2);

\draw [draw=black] (1.5,0) rectangle (3.5,1);
\draw [draw=black] (1.5,1) rectangle (3.5,2);
\draw [draw=black] (3.5,0) rectangle (4,1);
\draw [draw=black] (3.5,1) rectangle (4,2);

\draw [draw=black] (4,0) rectangle (5,1);
\draw [draw=black] (4,1) rectangle (5,2);
\draw [draw=black] (5,0) rectangle (5.5,1);
\draw [draw=black] (5,1) rectangle (5.5,2);
\draw [draw=black] (5.5,0) rectangle (6.5,1);
\draw [draw=black] (5.5,1) rectangle (6.5,2);
\draw [draw=black] (6.5,0) rectangle (7,1);
\draw [draw=black] (6.5,1) rectangle (7,2);
\draw [draw=black] (7,0) rectangle (9,1);
\draw [draw=black] (7,1) rectangle (9,2);
\draw [draw=black] (9,0) rectangle (9.5,1);
\draw [draw=black] (9,1) rectangle (9.5,2);
\draw [draw=black] (9.5,0) rectangle (10.5,1);
\draw [draw=black] (9.5,1) rectangle (10.5,2);

\node[] at (0.5,1.5) {$B_1^q$};
 \node[] at (0.5,0.5) {$B_2^q$};
  \node[] at (2,0.5) {$.$};
 \node[] at (2.5,0.5) {$.$};
 \node[] at (3,0.5) {$.$};
 \node[] at (2,1.5) {$.$};
 \node[] at (2.5,1.5) {$.$};
 \node[] at (3,1.5) {$.$};
 \node[] at (4.5,1.5) {$B_{2k-1}^q$};
 \node[] at (4.5,0.5) {$B_{2k}^q$};
 \node[] at (6,1.5) {$B_{1}^p$};
 \node[] at (6,0.5) {$B_2^p$};
 
 \node[] at (7.5,0.5) {$.$};
 \node[] at (8,0.5) {$.$};
 \node[] at (8.5,0.5) {$.$};
 \node[] at (7.5,1.5) {$.$};
 \node[] at (8,1.5) {$.$};
 \node[] at (8.5,1.5) {$.$};
 
 \node[] at (10,1.5) {$B_{2C-1}^p$};
 \node[] at (10,0.5) {$B_{2C}^p$};

    \draw (1,0) edge (1.5,0.5);
\draw (1.25,0) edge (1.5,0.25);
\draw (1,0.25) edge (1.5,0.75);
\draw (1,0.5) edge (1.5,1);
\draw (1,0.75) edge (1.25,1);
 
    \draw (1,1) edge (1.5,1.5);
\draw (1.25,1) edge (1.5,1.25);
\draw (1,1.25) edge (1.5,1.75);
\draw (1,1.5) edge (1.5,2);
\draw (1,1.75) edge (1.25,2);
 
  \draw (3.5,0) edge (4,0.5);
\draw (3.75,0) edge (4,0.25);
\draw (3.5,0.25) edge (4,0.75);
\draw (3.5,0.5) edge (4,1);
\draw (3.5,0.75) edge (3.75,1);

\draw (3.5,1) edge (4,1.5);
\draw (3.75,1) edge (4,1.25);
\draw (3.5,1.25) edge (4,1.75);
\draw (3.5,1.5) edge (4,2);
\draw (3.5,1.75) edge (3.75,2);
  
 \draw (5,0) edge (5.5,0.5);
\draw (5.25,0) edge (5.5,0.25);
\draw (5,0.25) edge (5.5,0.75);
\draw (5,0.5) edge (5.5,1);
\draw (5,0.75) edge (5.25,1);
 
  \draw (5,1) edge (5.5,1.5);
\draw (5.25,1) edge (5.5,1.25);
\draw (5,1.25) edge (5.5,1.75);
\draw (5,1.5) edge (5.5,2);
\draw (5,1.75) edge (5.25,2);

\draw (6.5,0) edge (7,0.5);
\draw (6.75,0) edge (7,0.25);
\draw (6.5,0.25) edge (7,0.75);
\draw (6.5,0.5) edge (7,1);
\draw (6.5,0.75) edge (6.75,1);
 
 \draw (6.5,1) edge (7,1.5);
\draw (6.75,1) edge (7,1.25);
\draw (6.5,1.25) edge (7,1.75);
\draw (6.5,1.5) edge (7,2);
\draw (6.5,1.75) edge (6.75,2);

\draw (9,0) edge (9.5,0.5);
\draw (9.25,0) edge (9.5,0.25);
\draw (9,0.25) edge (9.5,0.75);
\draw (9,0.5) edge (9.5,1);
\draw (9,0.75) edge (9.25,1);
 
 \draw (9,1) edge (9.5,1.5);
\draw (9.25,1) edge (9.5,1.25);
\draw (9,1.25) edge (9.5,1.75);
\draw (9,1.5) edge (9.5,2);
\draw (9,1.75) edge (9.25,2);
\end{tikzpicture}

\vspace{0.5cm}
\begin{tikzpicture}[scale=1]
\color{black}
\node[right] at (0,2.25) {- Type 2: $n'_q=2K$ and $n'_p=2C+1$ with $Z(\sigma)=Kq+(C+1)p+(K+C)s$};
\draw [draw=black] (0,0) rectangle (1,1);
\draw [draw=black] (0,1) rectangle (1,2);
\draw [draw=black] (1,0) rectangle (1.5,1);
\draw [draw=black] (1,1) rectangle (1.5,2);

\draw [draw=black] (1.5,0) rectangle (3.5,1);
\draw [draw=black] (1.5,1) rectangle (3.5,2);
\draw [draw=black] (3.5,0) rectangle (4,1);
\draw [draw=black] (3.5,1) rectangle (4,2);

\draw [draw=black] (4,0) rectangle (5,1);
\draw [draw=black] (4,1) rectangle (5,2);
\draw [draw=black] (5,0) rectangle (5.5,1);
\draw [draw=black] (5,1) rectangle (5.5,2);
\draw [draw=black] (5.5,0) rectangle (6.5,1);
\draw [draw=black] (5.5,1) rectangle (6.5,2);
\draw [draw=black] (6.5,0) rectangle (7,1);
\draw [draw=black] (6.5,1) rectangle (7,2);
\draw [draw=black] (7,0) rectangle (9,1);
\draw [draw=black] (7,1) rectangle (9,2);
\draw [draw=black] (9,0) rectangle (9.5,1);
\draw [draw=black] (9,1) rectangle (9.5,2);
\draw [draw=black] (9.5,0) rectangle (10.5,1);
\draw [draw=black] (9.5,1) rectangle (10.5,2);

\draw [draw=black] (10.5,1) rectangle (11,2);
\draw [draw=black] (11,1) rectangle (12,2);

\node[] at (0.5,1.5) {\color{black}$B_1^q$};
 \node[] at (0.5,0.5) {\color{black}$B_2^q$};
  \node[] at (2,0.5) {\color{black}$.$};
 \node[] at (2.5,0.5) {\color{black}$.$};
 \node[] at (3,0.5) {\color{black}$.$};
 \node[] at (2,1.5) {\color{black}$.$};
 \node[] at (2.5,1.5) {\color{black}$.$};
 \node[] at (3,1.5) {\color{black}$.$};
 \node[] at (4.5,1.5) {\color{black}$B_{2k-1}^q$};
 \node[] at (4.5,0.5) {\color{black}$B_{2k}^q$};
 \node[] at (6,1.5) {\color{black}$B_{1}^p$};
 \node[] at (6,0.5) {\color{black}$B_2^p$};
 
 \node[] at (7.5,0.5) {\color{black}$.$};
 \node[] at (8,0.5) {\color{black}$.$};
 \node[] at (8.5,0.5) {\color{black}$.$};
 \node[] at (7.5,1.5) {\color{black}$.$};
 \node[] at (8,1.5) {\color{black}$.$};
 \node[] at (8.5,1.5) {\color{black}$.$};
 
 \node[] at (10,1.5) {\color{black}$B_{2C-1}^p$};
 \node[] at (10,0.5) {\color{black}$B_{2C}^p$};
 \node[] at (11.5,1.5) {\color{black}$B_{2C+1}^p$};

\draw [draw=black](1,0) edge (1.5,0.5);
\draw [draw=black](1.25,0) edge (1.5,0.25);
\draw [draw=black](1,0.25) edge (1.5,0.75);
\draw [draw=black](1,0.5) edge (1.5,1);
\draw [draw=black](1,0.75) edge (1.25,1);
 
\draw [draw=black](1,1) edge (1.5,1.5);
\draw [draw=black](1.25,1) edge (1.5,1.25);
\draw [draw=black](1,1.25) edge (1.5,1.75);
\draw [draw=black](1,1.5) edge (1.5,2);
\draw [draw=black](1,1.75) edge (1.25,2);
 
  \draw [draw=black](3.5,0) edge (4,0.5);
\draw [draw=black](3.75,0) edge (4,0.25);
\draw [draw=black](3.5,0.25) edge (4,0.75);
\draw [draw=black](3.5,0.5) edge (4,1);
\draw [draw=black](3.5,0.75) edge (3.75,1);

\draw [draw=black](3.5,1) edge (4,1.5);
\draw [draw=black](3.75,1) edge (4,1.25);
\draw [draw=black] (3.5,1.25) edge (4,1.75);
\draw [draw=black](3.5,1.5) edge (4,2);
\draw [draw=black](3.5,1.75) edge (3.75,2);
  
 \draw [draw=black](5,0) edge (5.5,0.5);
\draw [draw=black](5.25,0) edge (5.5,0.25);
\draw [draw=black](5,0.25) edge (5.5,0.75);
\draw [draw=black](5,0.5) edge (5.5,1);
\draw [draw=black](5,0.75) edge (5.25,1);
 
\draw [draw=black](5,1) edge (5.5,1.5);
\draw [draw=black](5.25,1) edge (5.5,1.25);
\draw [draw=black](5,1.25) edge (5.5,1.75);
\draw [draw=black](5,1.5) edge (5.5,2);
\draw [draw=black](5,1.75) edge (5.25,2);

\draw [draw=black](6.5,0) edge (7,0.5);
\draw [draw=black](6.75,0) edge (7,0.25);
\draw [draw=black](6.5,0.25) edge (7,0.75);
\draw [draw=black](6.5,0.5) edge (7,1);
\draw [draw=black](6.5,0.75) edge (6.75,1);
 
 \draw [draw=black](6.5,1) edge (7,1.5);
\draw [draw=black](6.75,1) edge (7,1.25);
\draw [draw=black](6.5,1.25) edge (7,1.75);
\draw [draw=black](6.5,1.5) edge (7,2);
\draw [draw=black](6.5,1.75) edge (6.75,2);

\draw [draw=black](9,0) edge (9.5,0.5);
\draw [draw=black](9.25,0) edge (9.5,0.25);
\draw [draw=black](9,0.25) edge (9.5,0.75);
\draw [draw=black](9,0.5) edge (9.5,1);
\draw [draw=black](9,0.75) edge (9.25,1);
 
 \draw [draw=black](9,1) edge (9.5,1.5);
\draw [draw=black](9.25,1) edge (9.5,1.25);
\draw [draw=black](9,1.25) edge (9.5,1.75);
\draw [draw=black](9,1.5) edge (9.5,2);
\draw [draw=black](9,1.75) edge (9.25,2);

 \draw [draw=black](10.5,1) edge (11,1.5);
\draw [draw=black](10.75,1) edge (11,1.25);
\draw [draw=black](10.5,1.25) edge (11,1.75);
\draw [draw=black](10.5,1.5) edge (11,2);
\draw [draw=black](10.5,1.75) edge (10.75,2);
\end{tikzpicture}

\vspace{0.5cm}
\begin{tikzpicture}[scale=1]
\color{black}
\node[right] at (0,2.25) {- Type 3: $n'_q=2K+1$ and $n'_p=2C$ with $Z(\sigma)=Kq+(C+1)p+(K+C)s$};
\draw [draw=black] (0,0) rectangle (1,1);
\draw [draw=black] (0,1) rectangle (1,2);
\draw [draw=black] (1,0) rectangle (1.5,1);
\draw [draw=black] (1,1) rectangle (1.5,2);

\draw [draw=black] (1.5,0) rectangle (3.5,1);
\draw [draw=black] (1.5,1) rectangle (3.5,2);
\draw [draw=black] (3.5,0) rectangle (4,1);
\draw [draw=black] (3.5,1) rectangle (4,2);

\draw [draw=black] (4,0) rectangle (5,1);
\draw [draw=black] (4,1) rectangle (5,2);
\draw [draw=black] (5,0) rectangle (5.5,1);
\draw [draw=black] (5,1) rectangle (5.5,2);
\draw [draw=black] (5.5,0) rectangle (6,1);
\draw [draw=black] (6,0) rectangle (6.5,1);
\draw [draw=black] (6.5,0) rectangle (7,1);
\draw [draw=black] (5.5,1) rectangle (6.5,2);
\draw [draw=black] (6.5,1) rectangle (7,2);
\draw [draw=black] (7,0) rectangle (7.5,1);
\draw [draw=black] (7,1) rectangle (8,2);
\draw [draw=black] (7.5,0) rectangle (8.5,1);
\draw [draw=black] (8,1) rectangle (8.5,2);
\draw [draw=black] (8.5,0) rectangle (9,1);

\draw [draw=black] (8.5,1) rectangle (10.5,2);
\draw [draw=black] (9,0) rectangle (11,1);
\draw [draw=black] (10.5,1) rectangle (11,2);
\draw [draw=black] (11,0) rectangle (11.5,1);
\draw [draw=black] (11,1) rectangle (12,2);
\draw [draw=black] (11.5,0) rectangle (12.5,1);

\node[] at (0.5,1.5) {$B_1^q$};
 \node[] at (0.5,0.5) {$B_2^q$};
  \node[] at (2,0.5) {$.$};
 \node[] at (2.5,0.5) {$.$};
 \node[] at (3,0.5) {$.$};
 \node[] at (2,1.5) {$.$};
 \node[] at (2.5,1.5) {$.$};
 \node[] at (3,1.5) {$.$};
 \node[] at (4.5,1.5) {$B_{2k-1}^q$};
 \node[] at (4.5,0.5) {$B_{2k}^q$};
 \node[] at (6,1.5) {$B_{2k+1}^q$};
 \node[] at (7.5,1.5) {$B_{3}^p$};
 \node[] at (8,0.5) {$B_{4}^p$};
  \node[] at (5.75,0.5) {$B_{1}^p$};
\node[] at (6.75,0.5) {$B_{2}^p$};
\node[] at (11.5,1.5) {$B_{2C-1}^p$};
\node[] at (12,0.5) {$B_{2C}^p$};


\node[] at (9,1.5) {$.$};
 \node[] at (9.5,1.5) {$.$};
 \node[] at (10,1.5) {$.$};
 \node[] at (9.5,0.5) {$.$};
 \node[] at (10,0.5) {$.$};
\node[] at (10.5,0.5) {$.$};

    \draw (1,0) edge (1.5,0.5);
\draw (1.25,0) edge (1.5,0.25);
\draw (1,0.25) edge (1.5,0.75);
\draw (1,0.5) edge (1.5,1);
\draw (1,0.75) edge (1.25,1);
 
    \draw (1,1) edge (1.5,1.5);
\draw (1.25,1) edge (1.5,1.25);
\draw (1,1.25) edge (1.5,1.75);
\draw (1,1.5) edge (1.5,2);
\draw (1,1.75) edge (1.25,2);
 
  \draw (3.5,0) edge (4,0.5);
\draw (3.75,0) edge (4,0.25);
\draw (3.5,0.25) edge (4,0.75);
\draw (3.5,0.5) edge (4,1);
\draw (3.5,0.75) edge (3.75,1);

\draw (3.5,1) edge (4,1.5);
\draw (3.75,1) edge (4,1.25);
\draw (3.5,1.25) edge (4,1.75);
\draw (3.5,1.5) edge (4,2);
\draw (3.5,1.75) edge (3.75,2);
  
 \draw (5,0) edge (5.5,0.5);
\draw (5.25,0) edge (5.5,0.25);
\draw (5,0.25) edge (5.5,0.75);
\draw (5,0.5) edge (5.5,1);
\draw (5,0.75) edge (5.25,1);
 
  \draw (5,1) edge (5.5,1.5);
\draw (5.25,1) edge (5.5,1.25);
\draw (5,1.25) edge (5.5,1.75);
\draw (5,1.5) edge (5.5,2);
\draw (5,1.75) edge (5.25,2);
 
  \draw (6,0) edge (6.5,0.5);
\draw (6.25,0) edge (6.5,0.25);
\draw (6,0.25) edge (6.5,0.75);
\draw (6,0.5) edge (6.5,1);
\draw (6,0.75) edge (6.25,1);

\draw (7,0) edge (7.5,0.5);
\draw (7.25,0) edge (7.5,0.25);
\draw (7,0.25) edge (7.5,0.75);
\draw (7,0.5) edge (7.5,1);
\draw (7,0.75) edge (7.25,1);

\draw (6.5,1) edge (7,1.5);
\draw (6.75,1) edge (7,1.25);
\draw (6.5,1.25) edge (7,1.75);
\draw (6.5,1.5) edge (7,2);
\draw (6.5,1.75) edge (6.75,2);

\draw (8,1) edge (8.5,1.5);
\draw (8.25,1) edge (8.5,1.25);
\draw (8,1.25) edge (8.5,1.75);
\draw (8,1.5) edge (8.5,2);
\draw (8,1.75) edge (8.25,2);

\draw (8.5,0) edge (9,0.5);
\draw (8.75,0) edge (9,0.25);
\draw (8.5,0.25) edge (9,0.75);
\draw (8.5,0.5) edge (9,1);
\draw (8.5,0.75) edge (8.75,1);

\draw (10.5,1) edge (11,1.5);
\draw (10.75,1) edge (11,1.25);
\draw (10.5,1.25) edge (11,1.75);
\draw (10.5,1.5) edge (11,2);
\draw (10.5,1.75) edge (10.75,2);

\draw (11,0) edge (11.5,0.5);
\draw (11.25,0) edge (11.5,0.25);
\draw (11,0.25) edge (11.5,0.75);
\draw (11,0.5) edge (11.5,1);
\draw (11,0.75) edge (11.25,1);
\end{tikzpicture}

\vspace{0.5cm}
\begin{tikzpicture}[scale=1]
\color{black}

\node[right] at (0,2.25) {- Type 4: $n'_q=2K+1$ and $n'_p=2C+1$ with $Z(\sigma)=(K+1)q+Cp+(K+C)s$};
\draw [draw=black] (0,0) rectangle (1,1);
\draw [draw=black] (0,1) rectangle (1,2);
\draw [draw=black] (1,0) rectangle (1.5,1);
\draw [draw=black] (1,1) rectangle (1.5,2);

\draw [draw=black] (1.5,0) rectangle (3.5,1);
\draw [draw=black] (1.5,1) rectangle (3.5,2);
\draw [draw=black] (3.5,0) rectangle (4,1);
\draw [draw=black] (3.5,1) rectangle (4,2);

\draw [draw=black] (4,0) rectangle (5,1);
\draw [draw=black] (4,1) rectangle (5,2);
\draw [draw=black] (5,0) rectangle (5.5,1);
\draw [draw=black] (5,1) rectangle (5.5,2);
\draw [draw=black] (5.5,0) rectangle (6,1);
\draw [draw=black] (6,0) rectangle (6.5,1);
\draw [draw=black] (6.5,0) rectangle (7,1);
\draw [draw=black] (5.5,1) rectangle (6.5,2);
\draw [draw=black] (6.5,1) rectangle (7,2);
\draw [draw=black] (7,0) rectangle (7.5,1);
\draw [draw=black] (7,1) rectangle (8,2);
\draw [draw=black] (7.5,0) rectangle (8.5,1);
\draw [draw=black] (8,1) rectangle (8.5,2);
\draw [draw=black] (8.5,0) rectangle (9,1);

\draw [draw=black] (8.5,1) rectangle (10.5,2);
\draw [draw=black] (9,0) rectangle (11,1);
\draw [draw=black] (10.5,1) rectangle (11,2);
\draw [draw=black] (11,0) rectangle (11.5,1);
\draw [draw=black] (11,1) rectangle (12,2);
\draw [draw=black] (11.5,0) rectangle (12.5,1);
\draw [draw=black] (12,1) rectangle (12.5,2);
\draw [draw=black] (12.5,1) rectangle (13.5,2);

\node[] at (0.5,1.5) {$B_1^q$};
 \node[] at (0.5,0.5) {$B_2^q$};
  \node[] at (2,0.5) {$.$};
 \node[] at (2.5,0.5) {$.$};
 \node[] at (3,0.5) {$.$};
 \node[] at (2,1.5) {$.$};
 \node[] at (2.5,1.5) {$.$};
 \node[] at (3,1.5) {$.$};
 \node[] at (4.5,1.5) {$B_{2k-1}^q$};
 \node[] at (4.5,0.5) {$B_{2k}^q$};
 \node[] at (6,1.5) {$B_{2k+1}^q$};
 \node[] at (7.5,1.5) {$B_{3}^p$};
 \node[] at (8,0.5) {$B_{4}^p$};
  \node[] at (5.75,0.5) {$B_{1}^p$};
\node[] at (6.75,0.5) {$B_{2}^p$};
\node[] at (11.5,1.5) {$B_{2C-1}^p$};
\node[] at (12,0.5) {$B_{2C}^p$};
\node[] at (13,1.5) {$B_{2C+1}^p$};


\node[] at (9,1.5) {$.$};
 \node[] at (9.5,1.5) {$.$};
 \node[] at (10,1.5) {$.$};
 \node[] at (9.5,0.5) {$.$};
 \node[] at (10,0.5) {$.$};
\node[] at (10.5,0.5) {$.$};
 \draw (1,0) edge (1.5,0.5);
\draw (1.25,0) edge (1.5,0.25);
\draw (1,0.25) edge (1.5,0.75);
\draw (1,0.5) edge (1.5,1);
\draw (1,0.75) edge (1.25,1);
 
    \draw (1,1) edge (1.5,1.5);
\draw (1.25,1) edge (1.5,1.25);
\draw (1,1.25) edge (1.5,1.75);
\draw (1,1.5) edge (1.5,2);
\draw (1,1.75) edge (1.25,2);
 
  \draw (3.5,0) edge (4,0.5);
\draw (3.75,0) edge (4,0.25);
\draw (3.5,0.25) edge (4,0.75);
\draw (3.5,0.5) edge (4,1);
\draw (3.5,0.75) edge (3.75,1);

\draw (3.5,1) edge (4,1.5);
\draw (3.75,1) edge (4,1.25);
\draw (3.5,1.25) edge (4,1.75);
\draw (3.5,1.5) edge (4,2);
\draw (3.5,1.75) edge (3.75,2);
  
 \draw (5,0) edge (5.5,0.5);
\draw (5.25,0) edge (5.5,0.25);
\draw (5,0.25) edge (5.5,0.75);
\draw (5,0.5) edge (5.5,1);
\draw (5,0.75) edge (5.25,1);
 
  \draw (5,1) edge (5.5,1.5);
\draw (5.25,1) edge (5.5,1.25);
\draw (5,1.25) edge (5.5,1.75);
\draw (5,1.5) edge (5.5,2);
\draw (5,1.75) edge (5.25,2);
 
  \draw (6,0) edge (6.5,0.5);
\draw (6.25,0) edge (6.5,0.25);
\draw (6,0.25) edge (6.5,0.75);
\draw (6,0.5) edge (6.5,1);
\draw (6,0.75) edge (6.25,1);

\draw (7,0) edge (7.5,0.5);
\draw (7.25,0) edge (7.5,0.25);
\draw (7,0.25) edge (7.5,0.75);
\draw (7,0.5) edge (7.5,1);
\draw (7,0.75) edge (7.25,1);

\draw (6.5,1) edge (7,1.5);
\draw (6.75,1) edge (7,1.25);
\draw (6.5,1.25) edge (7,1.75);
\draw (6.5,1.5) edge (7,2);
\draw (6.5,1.75) edge (6.75,2);

\draw (8,1) edge (8.5,1.5);
\draw (8.25,1) edge (8.5,1.25);
\draw (8,1.25) edge (8.5,1.75);
\draw (8,1.5) edge (8.5,2);
\draw (8,1.75) edge (8.25,2);

\draw (8.5,0) edge (9,0.5);
\draw (8.75,0) edge (9,0.25);
\draw (8.5,0.25) edge (9,0.75);
\draw (8.5,0.5) edge (9,1);
\draw (8.5,0.75) edge (8.75,1);

\draw (10.5,1) edge (11,1.5);
\draw (10.75,1) edge (11,1.25);
\draw (10.5,1.25) edge (11,1.75);
\draw (10.5,1.5) edge (11,2);
\draw (10.5,1.75) edge (10.75,2);

\draw (11,0) edge (11.5,0.5);
\draw (11.25,0) edge (11.5,0.25);
\draw (11,0.25) edge (11.5,0.75);
\draw (11,0.5) edge (11.5,1);
\draw (11,0.75) edge (11.25,1);

\draw (12,1) edge (12.5,1.5);
\draw (12.25,1) edge (12.5,1.25);
\draw (12,1.25) edge (12.5,1.75);
\draw (12,1.5) edge (12.5,2);
\draw (12,1.75) edge (12.25,2);
\end{tikzpicture}
    \caption{Types of scheduling.}\label{ft}
\end{center}
\end{figure}

Consider the possible values of $\sigma_1$ regarding the parity of $n_p$, $n_q$ and $a_1$, and we show that $Z(\sigma)> Z(\sigma_1)$:
\begin{itemize}
\item \textbf{$n_q$ is even and $n_p$ is odd:}
\begin{itemize}
    \item If $a_1$ is even, $\sigma_1$ is of type 1 with   $Z(\sigma_1)=\left(\dfrac{2K-a_1}{2}\right)q+\left(\dfrac{2C+a_1+2}{2}\right)p+(K+C)s$
    $<Kq+(C+1)p+(K+C)s$ (because $\dfrac{a_1}{2}q > \dfrac{a_1}{2} p$), so $Z(\sigma)> Z(\sigma_1)$.
    \item If $a_1$ is odd, $\sigma_1$ is of type 4 with  $Z(\sigma_1)=\left(\dfrac{2K-a_1+1}{2}\right)q+\left(\dfrac{2C+a_1+1}{2}\right)p+(K+C)s$ $<Kq+(C+1)p+(K+C)s$ (because  $\left(\dfrac{a_1-1}{2}\right)q > \left(\dfrac{a_1-1}{2}\right)p$), so $Z(\sigma)> Z(\sigma_1)$.
\end{itemize}
\item \textbf{$n_q$ is odd and $n_p$ is even:}
\begin{itemize}
    \item If $a_1$ is even, $\sigma_1$ is of type 4 with  $Z(\sigma_1)=\left(\dfrac{2K-a_1+2}{2}\right)q+\left(\dfrac{2C+a_1}{2}\right)p+(K+C)s$
    $<Kq+(C+1)p+(K+C)s$ (because $\left(\dfrac{a_1-2}{2}\right)q > \left(\dfrac{a_1-2}{2}\right) p$), so $Z(\sigma)> Z(\sigma_1)$.
    \item If $a_1$ is odd, $\sigma_1$ is of type 1 with  $Z(\sigma_1)=\left(\dfrac{2K-a_1+1}{2}\right)q+\left(\dfrac{2C+a_1+1}{2}\right)p+(K+C)s$
     $<Kq+(C+1)p+(K+C)s$ (because $\left(\dfrac{a_1-1}{2}\right)q > \left(\dfrac{a_1-1}{2}\right)p$), so $Z(\sigma)> Z(\sigma_1)$.
\end{itemize}
\end{itemize}

So $\forall a_1$, $Z(\sigma)> Z(\sigma_1)$. Therefore, $\sigma$ is not an optimal solution of the problem.
\end{proof}

\begin{lemma}
If we apply Algorithm \ref{Tux} a second time to $\sigma_1$, then the makespan increases.
\end{lemma}
   
\begin{proof}
If one between $n_q$ and $n_{p}$ is even and the other is odd, ($\sigma_1$ be the solution found at iteration $1$ of Algorithm \ref{Tux}, such that, we delete from $\sigma$, $a_1$ batches of processing time $q$ 
 and we add $(a_1+1)$ batches of processing time $p$), so, for all value of $a_1$, the schedule $\sigma_1$ is the type 1 or the type 4. Let $\sigma_m$ be the solution found if we apply the Algorithm \ref{Tux},  $(m-1)$ times to the solution $\sigma_1$.  So, we delete from $\sigma$, $\alpha=\sum_{i=1}^{m} a_i$ batches of processing time $q$ and we add $(\alpha+m)$ batches of   processing time $p$.  So, for all value of $\alpha$ and $m$, the schedule $\sigma_m$ is of type 1 or type 2 or type 3 or type 4.

 We have, $2a_i p+2a_i s\geq a_iq+(a_i-1)p+(2a_i-1)s$ for $i=1,\ldots,m$. Thus, $(a_i+1)p+s\geq a_iq$. Then, $\sum_{i=1}^m ((a_i+1)p+s)\geq  \sum_{i=1}^m a_iq$. So, $(\alpha+m)p+ms\geq\alpha q$.\\

The possible values for  $\sigma_m$ are:
 
\begin{itemize}
    \item \textbf{If $n_q$ is even and $n_p$ is odd:} if $\sigma_m$ is of
    \begin{itemize}
        \item type 1:  $Z(\sigma_m)=\left(\dfrac{2K-\alpha}{2}\right)q+\left(\dfrac{2C+\alpha+m+1}{2}\right)p+\left(K+C+\dfrac{m-1}{2}\right)s$.
        \item type 2:
        $Z(\sigma_m)=\left(\dfrac{2K-\alpha}{2}\right)q+\left(\dfrac{2C+\alpha+m+2}{2}\right)p+\left(K+C+\dfrac{m}{2}\right)s$.
        \item type 3:
        $Z(\sigma_m)=\left(\dfrac{2K-\alpha-1}{2}\right)q+\left(\dfrac{2C+\alpha+m+3}{2}\right)p+\left(K+C+\dfrac{m}{2}\right)s$.
        \item type 4:
        $Z(\sigma_m)=\left(\dfrac{2K-\alpha+1}{2}\right)q+\left(\dfrac{2C+\alpha+m}{2}\right)p+\left(K+C+\dfrac{m-1}{2}\right)s$.
        
    \end{itemize}
    \item \textbf{If $n_q$ is odd and $n_p$ is even:} If $\sigma_m$ is of
    \begin{itemize}
        \item type 1: $Z(\sigma_m)=\left(\dfrac{2K-\alpha+1}{2}\right)q+\left(\dfrac{2C+\alpha+m}{2}\right)p+\left(K+C+\dfrac{m-1}{2}\right)s$.
        \item type 2: $Z(\sigma_m)=\left(\dfrac{2K-\alpha+1}{2}\right)q+\left(\dfrac{2C+\alpha+m+1}{2}\right)p+\left(K+C+\dfrac{m}{2}\right)s$.
        \item type 3: $Z(\sigma_m)=\left(\dfrac{2K-\alpha}{2}\right)q+\left(\dfrac{2C+\alpha+m+2}{2}\right)p+\left(K+C+\dfrac{m}{2}\right)s$. 
        \item type 4:  $Z(\sigma_m)=\left(\dfrac{2K-\alpha+2}{2}\right)q+\left(\dfrac{2C+\alpha+m-1}{2}\right)p+$\\ \hspace*{4.5cm} $\left(K+C+\dfrac{m-1}{2}\right)s$.
    \end{itemize}
\end{itemize}
 
To show that $Z(\sigma_1)\leq Z(\sigma_m)$,  we check all the possible cases:
 
\begin{itemize}    
 \item \textbf{If $n_q$ is even and $n_p$ is odd:}
\begin{itemize}
\item[•]\textbf{$\sigma_1$ is of type 1 and  $\sigma_m$ is of type 1} (subcase A11)\\

$\left(\dfrac{2K-a_1}{2}\right)q+\left(\dfrac{2C+a_1+2}{2}\right)p+(K+C)s\leq\left(\dfrac{2K-\alpha}{2}\right)q+\left(\dfrac{2C+\alpha+m+1}{2}\right)p+\left(K+C+\dfrac{m-1}{2}\right)s$ if, and only if,  $\left(\dfrac{a_1-\alpha-m+1}{2}\right)p-\left(\dfrac{m-1}{2}\right)s\leq\left(\dfrac{a_1-\alpha}{2}\right)q$.
 From the hypothesis:  $\sum_{i=1}^{m-1}(( a_i+1)p+s)\geq \sum_{i=1}^{m-1} a_i q$, if we apply Algorithm \ref{Tux}  $(m-2)$ times to   $\sigma_1$, $(\alpha-a_m+m-1)p+(m-1)s\geq(\alpha-a_m) q$ and we have   $Min_{i=1..m}\lbrace a_i\rbrace =a_m$,  then $(a_m-a_1)p>(a_m-a_1)q$, so ($\alpha-a_m+m-1)p+(a_m-a_1)p+(m-1)s>(\alpha-a_m) q+(a_m-a_1)q$ we obtain  $\left(\dfrac{a_1-\alpha-m+1}{2}\right)p-\left(\dfrac{m-1}{2}\right)s<\left(\dfrac{a_1-\alpha}{2}\right) q$, so $Z(\sigma_1)< Z(\sigma_m)$.\\

\item[•]\textbf{$\sigma_1$ is of type 1 and  $\sigma_m$ is of type 2}  (subcase A12)\\

 $\left(\dfrac{2K-a_1}{2}\right)q+\left(\dfrac{2C+a_1+2}{2}\right)p+(K+C)s\leq\left(\dfrac{2K-\alpha}{2}\right)q+\left(\dfrac{2C+\alpha+m+2}{2}\right)p$ $+\left(K+C+\dfrac{m}{2}\right)s$ if, and only if,  $\left(\dfrac{a_1-\alpha-m}{2}\right)p-\dfrac{m}{2}s\leq \left(\dfrac{a_1-\alpha}{2}\right)q$.
From the hypothesis:  $(\alpha+m)p+ms\geq\alpha q$, then $(\alpha+m-a_1)p+ms>(\alpha-a_1) q$, because $-a_1 p>-a_1 q$, then \\$\left(\dfrac{a_1-\alpha-m}{2}\right)p-\dfrac{m}{2}s<\left(\dfrac{a_1-\alpha}{2}\right) q$, so $Z(\sigma_1)<Z(\sigma_m)$.\\

\item[•]\textbf{$\sigma_1$ is of type 1 and  $\sigma_m$ is of type 3} (subcase A13)\\

$\left(\dfrac{2K-a_1}{2}\right)q+\left(\dfrac{2C+a_1+2}{2}\right)p+(K+C)s\leq\left(\dfrac{2K-\alpha-1}{2}\right)q+$\\ {\flushright $\left(\dfrac{2C+\alpha+m+3}{2}\right)p+\left(K+C+\dfrac{m}{2}\right)s$\\} if, and only if,  $\left(\dfrac{a_1-\alpha-m-1}{2}\right)p-\left(\dfrac{m}{2}\right)s\leq\left(\dfrac{a_1-\alpha-1}{2}\right)q$.
From the  hypothesis:  $(\alpha+m)p+ms\geq \alpha q$, then  $\left(\dfrac{a_1-\alpha-m-1}{2}\right)p-\left(\dfrac{m}{2}\right)s<\left(\dfrac{a_1-\alpha-1}{2}\right) q$,  because $(1-a_1)p>(1-a_1)q$, so $Z(\sigma_1)< Z(\sigma_m)$.\\

\item[•]\textbf{$\sigma_1$ is of type 1 and  $\sigma_m$ is of type 4} (subcase A14)\\

 $\left(\dfrac{2K-a_1}{2}\right)q+\left(\dfrac{2C+a_1+2}{2}\right)p+\left(K+C\right)s$ $\leq\left(\dfrac{2K-\alpha+1}{2}\right)q+\left(\dfrac{2C+\alpha+m}{2}\right)p$ {\flushright $+\left(K+C+\left(\dfrac{m-1}{2}\right)\right)s$\\} if, and only if,  $\left(\dfrac{a_1-\alpha-m+2}{2}\right)p-\left(\dfrac{m-1}{2}\right)s\leq\left(\dfrac{a_1-\alpha+1}{2}\right)q$.\\
We have shown that $\left(\dfrac{a_1-\alpha-m}{2}\right)p-\left(\dfrac{m-1}{2}\right)s<\left(\dfrac{a_1-\alpha-1}{2}\right)q$, then \\
$\left(\dfrac{a_1-\alpha-m}{2}\right)p+p-\left(\dfrac{m-1}{2}\right)s<\left(\dfrac{a_1-\alpha-1}{2}\right)q+q$, so $Z(\sigma_1)<Z(\sigma_m)$.\\

\item[•]\textbf{$\sigma_1$ is of type 4 and  $\sigma_m$ is of type 1} (subcase A41)\\

$\left(\dfrac{2K-a_1+1}{2}\right)q+\left(\dfrac{2C+a_1+1}{2}\right)p+\left(K+C\right)s$ $\leq\left(\dfrac{2K-\alpha}{2}\right)q+$\\ {\flushright $\left(\dfrac{2C+\alpha+m+1}{2}\right)p+\left(K+C+\left(\dfrac{m-1}{2}\right)\right)s$\\} if, and only if,  $\left(\dfrac{a_1-\alpha-m}{2}\right)p-\left(\dfrac{m-1}{2}\right)s\leq\left(\dfrac{a_1-\alpha-1}{2}\right)q$.\\
  From the hypothesis, $(\alpha-a_m+m-1)p+(m-1)s\geq(\alpha-a_m) q$, then $(\alpha-a_m+m-1)p+(a_m-a_1+1)p+(m-1)s>(\alpha-a_m) q+(a_m-a_1+1)q$, because $(a_m-a_1+1)p>(a_m-a_1+1)q$, so $Z(\sigma_1)<Z(\sigma_m)$.\\

\item[•]\textbf{$\sigma_1$ is of type 4 and $\sigma_m$ is of type 2} (subcase A42), similar to subcase A13.

\item[•]\textbf{$\sigma_1$ is of type 4 and  $\sigma_m$ is of type 3} (subcase A43)\\

  $\left(\dfrac{2K-a_1+1}{2}\right)q+\left(\dfrac{2C+a_1+1}{2}\right)p+\left(K+C\right)s$  $\leq\left(\dfrac{2K-\alpha-1}{2}\right)q+$\\  {\flushright $\left(\dfrac{2C+\alpha+m+3}{2}\right)p+\left(K+C+\dfrac{m}{2}\right)s$\\} if, and only if,  $\left(\dfrac{a_1-\alpha-m-2}{2}\right)p-\left(\dfrac{m}{2}\right)s\leq\left(\dfrac{a_1-\alpha-2}{2}\right)q$.\\
From the hypothesis:  $(\alpha+m)p+ms\geq\alpha q$, then $(\alpha+m)p+(2-a_1)p+ms>\alpha q+(2-a_1)q$, because $(2-a_1)p>(2-a_1)q$, so $Z(\sigma_1)< Z(\sigma_m)$.

\item[•]\textbf{$\sigma_1$ is of type 4 and  $\sigma_m$ is of type 4} (subcase A44), similar to subcase A11.
\end{itemize}
\item \textbf{If $n_q$ is odd and $n_p$ is even:}
\begin{itemize}

\item[•]\textbf{$\sigma_1$ is of type 1 and  $\sigma_m$ is of type 1}, similar to subcase A11.

\item[•]\textbf{$\sigma_1$ is of type 1 and  $\sigma_m$ is of type 2}, similar to subcase A12.

\item[•]\textbf{$\sigma_1$ is of type 1 and  $\sigma_m$ is of type 3}, similar to subcase A13.

\item[•]\textbf{$\sigma_1$ is of type 1 and  $\sigma_m$ is of type 4}, similar to subcase A14.

\item[•]\textbf{$\sigma_1$ is of type 4 and  $\sigma_m$ is of type 1}, similar to subcase A41.
    
\item[•]\textbf{$\sigma_1$ is of type 4 and  $\sigma_m$ is of type 2}, similar to subcase A13.
  
\item[•]\textbf{$\sigma_1$ is of type 4 and  $\sigma_m$ is of type 3}, similar to subcase A43.

\item[•]\textbf{$\sigma_1$ is of type 4 and  $\sigma_m$ is of type 4}, similar to subcase A11.
\end{itemize}
\end{itemize}
\end{proof}

From Lemma 1 and Lemma 2, we conclude that Algorithm \ref{Tux} gives an optimal solution. As the three nested loops require a time in $O(n^6)$ and the matching algorithm a time in $O(n^3)$, then Algorithm \ref{Tux} requires a time in $O(n^9)$ but this complexity, calculated in the worst case, can be reduced. 
\end{proof}

\section{Mathematical programming model}

In this section, we present a mixed integer programming model (MILP) to solve the
scheduling problem $Bm, max|G=(V,E),b=2, s|C_{max}$.

Let $pb_{ij}=max\lbrace p_i, p_j\rbrace$ if the two jobs $J_i$ and $J_j$ are compatible and let $A$ be  the adjacency matrix of the graph $G=(V, E)$ such that:\\

$a_{ij}= \left\{\begin{tabular}{ll}
1  &  if the jobs  $J_i$  and $J_j$ are compatible,  $i,j = 1,\ldots,n$\\ 
0  & otherwise\\
\end{tabular} \right.$
\\

\textbf{Decision variables:} Define two types of variables,\\
 
$x_{ijk}= \left\{\begin{tabular}{ll}
1  &  if the jobs $J_i$ and $J_j$ are processed in a same batch on machine $M_k$,\\ 
0  & otherwise \\
\end{tabular} \right. $

$i<j=1,\ldots,n$  and $ k=1,\ldots,m$.\\

$ y_{ik}= \left\{\begin{tabular}{ll}
1  &  if the job  $J_i$ is schedule on machine $M_k$ as a single batch of jobs,\\
0  & otherwise \\
\end{tabular} \right.$

$i=1\ldots,n$ and $k=1,\ldots,m$.\\
 
 \textbf{Linear model:} The linear model of the problem $Bm/G=(V,E),b=2, s/C_{max}$ is:\\

$minimize$  $C_{max}$

subject to 

$\left\lbrace \begin{tabular}{lll}
$\sum_{k=1}^{m}\sum_{j>i}^{n} x_{ijk}\leq 1$ & $i=1,\ldots,n$ & (1) \\
& & \\
$\sum_{k=1}^{m} x_{ijk} \leq a_{ij}$ & $i<j=1,\ldots,n$ & (2) \\
& & \\
$\sum_{k=1}^{m} y_{ik} \leq 1$ & $i=1,\ldots,n$ & (3)  \\
& & \\
$\sum_{k=1}^{m} \sum_{j>i}^{n} x_{ijk} = 1-\sum_{k=1}^{m} y_{ik}$ & $i=1,\ldots,n$ & (4)  \\
& & \\
$\sum_{i=1}^{n} \sum_{j>i}^{n} pb_{ij} x_{ijk}+\sum_{i=1}^{n} p_{i} y_{ik}$ & & \\
$+ s(\sum_{i=1}^{n} \sum_{j>i}^{n} x_{ijk}+\sum_{i=1}^{n}y_{ik}-1)\leq C_{max}$ & $k=1,\ldots,m$ & (5) \\
& & \\
$C_{max}\geq0$ & &  (6) \\
& & \\
$x_{ijk}\in\lbrace0,1\rbrace$ , $y_{ik}\in\lbrace0,1\rbrace$ & $i<j=1,\ldots,n$\\ & and $k=1,\ldots,m$ & (7)
\end{tabular}\right.$\\

Where
\begin{itemize}
\item[•] Equations (1): indicate that each job is, at most, in one batch of two jobs on a single machine.
\item[•]  Equations (2): ensure that every two compatible jobs are, at most,  in one batch of two jobs on a single machine.
\item[•] Equations (3): ensure that each single job batch is in at most one machine.
\item[•] Equations (4): indicate that each job is in one batch on one machine.
\item[•] $\sum_{i=1}^{n}\sum_{j>i}^{n} pb_{ij} x_{ijk}$: the sum of the batch processing times with two jobs on each machine.
\item[•] $\sum_{i=1}^{n}p_{i} y_{ik}$: the sum of the batch processing times with a single job on each machine.
\item[•] $s(\sum_{i=1}^{n}\sum_{j>i}^{n} x_{ijk}+\sum_{i=1}^{n}y_{ik}-1)$: the sum of the batch preparation times on each machine.
\item[•] Equations (6): ensure that the corresponding decision variable must be positive.
\item[•] Equations (7): ensure that the corresponding decision variables are binary.\\
\end{itemize}

The number of variables and the number of constraints of a linear mathematical model
are two indices by which we can measure the dimension and the efficiency of the given model. The number of variables is $m n(\dfrac{n+1}{2})$ binary variables and one continuous variable. The total number of constraints is $(n(\frac{n+5}{2})+m)$.

\subsection*{Performance of the linear formulation}

We study in this section the performance of the linear formulation that was tested with IBM ILOG CPLEX solver (20.1.0). For different values of  $n \in \lbrace 10, 20, 30, 40, 50, 60, 70, 80, 90, 100, 200\rbrace$, we considered 50 of randomly generated instances solved in less than 900 $s$ for each density (low, medium and high)  of the graph $G$. 

\begin{figure}[h!]
\begin{minipage}{0.49\linewidth}
\begin{tikzpicture}[x=0.7mm,y=1.1mm,scale=1.5]
\color{black}

\draw[->,>=latex] (0, 0) -- (56, 0) node[below]{};
\draw[->,>=latex] (0, 0) -- (0,22) node[left]{};

\draw plot[smooth,very thick,xscale=0.5,yscale=1]coordinates{(0,2)(12.5,2.586)(25,3.8) (50,7.5)(75,15.697) (100,18.318)};

\node[scale=0.8] at (-2, 2) {2};
\node[scale=1] at (0, 2) {-};
 
 \node[scale=0.8] at (-2, 4) {4};
\node[scale=1] at (0, 4) {-};

\node[scale=0.8] at (-2, 6) {6};
\node[scale=1] at (0, 6) {-};

\node[scale=0.8] at (-2, 8) {8};
\node[scale=1] at (0, 8) {-};

\node[scale=0.8] at (-2, 10) {10};
\node[scale=1] at (0, 10) {-};

\node[scale=0.8] at (-2, 12) {12};
\node[scale=1] at (0, 12) {-};

\node[scale=0.8] at (-2, 14) {14};
\node[scale=1] at (0, 14) {-};

\node[scale=0.8] at (-2, 16) {16};
\node[scale=1] at (0, 16) {-};

\node[scale=0.8] at (-2, 18) {18};
\node[scale=1] at (0, 18) {-};

\node[scale=0.8] at (-2, 20) {20};
\node[scale=1] at (0, 20) {-};
\node[scale=0.8] at (0, 23) {time (s)};

\node[scale=0.8] at (60, -1) {density $(\%)$};
\node[scale=0.8] at (50, -2) {100};
\node[scale=0.4] at (50, 0) {|};
\node[scale=0.8] at (25, -2) {50};
\node[scale=0.4] at (25, 0) {|};
\node[scale=0.8] at (12.5, -2) {25};
\node[scale=0.4] at (12.5, 0) {|};
\node[scale=0.8] at (6.25,-2) {12.5};
\node[scale=0.4] at (6.25,0) {|};
\node[scale=0.8] at (37.5,-2) {75};
\node[scale=0.4] at (37.5,0) {|};
\end{tikzpicture}
\caption{Performance of the model with $n=50$.}\label{fig1}
\end{minipage}
\begin{minipage}{0.49\linewidth}
\begin{tikzpicture}[x=0.7mm,y=1.1mm,scale=1.5]
\color{black}

\draw[ ->,>=latex] (0, 0) -- (53, 0) node[below]{};
\draw[->,>=latex] (0, 0) -- (0,24) node[left]{};

\draw plot[smooth,very thick,xscale=0.5,yscale=1]coordinates{(0,0)(10,1.06)(20,3.1) (30,3) (50,3.8)(60,6.22) (70,10.58) (80,9.52) (90,11.25) (100,21.37)};

\node[scale=0.8] at (-2, 2) {2};
\node[scale=1] at (0, 2) {-};
 
 \node[scale=0.8] at (-2, 4) {4};
\node[scale=1] at (0, 4) {-};

\node[scale=0.8] at (-2, 6) {6};
\node[scale=1] at (0, 6) {-};

\node[scale=0.8] at (-2, 8) {8};
\node[scale=1] at (0, 8) {-};

\node[scale=0.8] at (-2, 10) {10};
\node[scale=1] at (0, 10) {-};

\node[scale=0.8] at (-2, 12) {12};
\node[scale=1] at (0, 12) {-};

\node[scale=0.8] at (-2, 14) {14};
\node[scale=1] at (0, 14) {-};

\node[scale=0.8] at (-2, 16) {16};
\node[scale=1] at (0, 16) {-};

\node[scale=0.8] at (-2, 18) {18};
\node[scale=1] at (0, 18) {-};

\node[scale=0.8] at (-2, 20) {20};
\node[scale=1] at (0, 20) {-};

\node[scale=0.8] at (-2, 22) {22};
\node[scale=1] at (0, 22) {-};
\node[scale=0.8] at (0, 25) {time (s)};
\node[scale=0.8] at (5, -2) {10};
\node[scale=0.4] at (5, 0) {|};

\node[scale=0.8] at (10, -2) {20};
\node[scale=0.4] at (10, 0) {|};

\node[scale=0.8] at (15, -2) {30};
\node[scale=0.4] at (15, 0) {|};

\node[scale=0.8] at (20,-2) {40};
\node[scale=0.4] at (20,0) {|};

\node[scale=0.8] at (25,-2) {50};
\node[scale=0.4] at (25,0) {|};

\node[scale=0.8] at (30,-2) {60};
\node[scale=0.4] at (30,0) {|};

\node[scale=0.8] at (35,-2) {70};
\node[scale=0.4] at (35,0) {|};

\node[scale=0.8] at (40,-2) {80};
\node[scale=0.4] at (40,0) {|};

\node[scale=0.8] at (45,-2) {90};
\node[scale=0.4] at (45,0) {|};

\node[scale=0.8] at (50,-2) {100};
\node[scale=0.4] at (50,0) {|};

\node[scale=0.8] at (53,-1) {$n$};

\end{tikzpicture}
\caption{Performance of the model with low density.}\label{fig2}
\end{minipage}
\end{figure}

\begin{figure}[h!]
\begin{minipage}{0.49\linewidth}
\begin{tikzpicture}[x=0.7mm,y=0.9mm,scale=1.5]
\color{black}
\draw[->,>=latex] (0, 0) -- (52, 0) node[below]{};
\draw[->,>=latex] (0, 0) -- (0,38) node[left]{};

\draw plot[smooth,very thick,xscale=0.25,yscale=0.5]coordinates{(0,0)(10,1.37)(20,4.1) (30,3.57) (40,4.94)(50,7.18)(60,9.95) (70,21.9) (80,22.79) (90,19.76) (100,17.51) (200,69.17)};

\node[scale=1] at (0,0.5*5) {-};

\node[scale=0.8] at (-2,0.5*10) { 10};
\node[scale=1] at (0,0.5*10) {-};

\node[scale=1] at (0,0.5*15) {-};

\node[scale=0.8] at (-2,0.5*20) {20};
\node[scale=1] at (0,0.5*20) {-};

\node[scale=1] at (0,0.5*25) {-};

\node[scale=0.8] at (-2,0.5*30) {30};
\node[scale=1] at (0,0.5*30) {-};

\node[scale=1] at (0,0.5*35) {-};

\node[scale=0.8] at (-2,0.5*40) {40};
\node[scale=1] at (0,0.5*40) {-};

\node[scale=1] at (0,0.5*45) {-};

\node[scale=0.8] at (-2,0.5*50) {50};
\node[scale=1] at (0,0.5*50) {-};

\node[scale=1] at (0,0.5*55) {-};

\node[scale=0.8] at (-2,0.5*60) {60};
\node[scale=1] at (0,0.5*60) {-};

\node[scale=1] at (0,0.5*65) {-};

\node[scale=0.8] at (-2,0.5*70) {70};
\node[scale=1] at (0,0.5*70) {-};
\node[scale=0.8] at (0,0.5*78) {time (s)};

\node[scale=0.8] at (0.25*10, -2) {10};
\node[scale=0.4] at (0.25*10, 0) {|};

\node[scale=0.4] at (0.25*20, 0) {|};

\node[scale=0.4] at (0.25*30, 0) {|};

\node[scale=0.4] at (0.25*40,0) {|};

\node[scale=0.8] at (0.25*50,-2) {50};
\node[scale=0.4] at (0.25*50,0) {|};

\node[scale=0.4] at (0.25*60,0) {|};

\node[scale=0.4] at (0.25*70,0) {|};

\node[scale=0.4] at (0.25*80,0) {|};

\node[scale=0.4] at (0.25*90,0) {|};

\node[scale=0.8] at (0.25*101,-2) {100};
\node[scale=0.4] at (0.25*100,0) {|};

\node[scale=0.8] at (0.25*200,-2) {200};
\node[scale=0.4] at (0.25*200,0) {|};
\node[scale=0.8] at (0.265*200,-1) {$n$};
\end{tikzpicture}
\caption{Performance of the model with medium density.}\label{fig3}
\end{minipage}
\begin{minipage}{0.49\linewidth}
\begin{tikzpicture}[x=0.7mm,y=0.9mm,scale=1.5]
\color{black}
\draw[->,>=latex] (0, 0) -- (53, 0) node[below]{};
\draw[->,>=latex] (0, 0) -- (0,39) node[left]{};

\draw plot[smooth,very thick,xscale=0.5,yscale=0.37]coordinates{(0,0)(10,1.53)(20,5.1) (30,4.06) (40,6.84)(50,18.31)(60,20.89) (70,39.54) (80,59.54) (90,80.41) (100,100)};

\node[scale=1] at (0,0.37*5) {-};

\node[scale=0.8] at (-2,0.37*10) { 10};
\node[scale=1] at (0,0.37*10) {-};

\node[scale=1] at (0,0.37*15) {-};

\node[scale=0.8] at (-2,0.37*20) {20};
\node[scale=1] at (0,0.37*20) {-};

\node[scale=1] at (0,0.37*25) {-};

\node[scale=0.8] at (-2,0.37*30) {30};
\node[scale=1] at (0,0.37*30) {-};

\node[scale=1] at (0,0.37*35) {-};

\node[scale=0.8] at (-2,0.37*40) {40};
\node[scale=1] at (0,0.37*40) {-};

\node[scale=1] at (0,0.37*45) {-};

\node[scale=0.8] at (-2,0.37*50) {50};
\node[scale=1] at (0,0.37*50) {-};

\node[scale=1] at (0,0.37*55) {-};

\node[scale=0.8] at (-2,0.37*60) {60};
\node[scale=1] at (0,0.37*60) {-};

\node[scale=1] at (0,0.37*65) {-};

\node[scale=0.8] at (-2,0.37*70) {70};
\node[scale=1] at (0,0.37*70) {-};

\node[scale=1] at (0,0.37*75) {-};

\node[scale=0.8] at (-2,0.37*80) {80};
\node[scale=1] at (0,0.37*80) {-};

\node[scale=1] at (0,0.37*85) {-};

\node[scale=0.8] at (-2,0.37*90) {90};
\node[scale=1] at (0,0.37*90) {-};

\node[scale=1] at (0,0.37*95) {-};

\node[scale=0.8] at (-2.4,0.37*100) {100};
\node[scale=1] at (0,0.37*100) {-};

\node[scale=0.8] at (0,0.40*100) {time (s)};

\node[scale=0.8] at (5, -2) {10};
\node[scale=0.4] at (5, 0) {|};

\node[scale=0.8] at (10, -2) {20};
\node[scale=0.4] at (10, 0) {|};

\node[scale=0.8] at (15, -2) {30};
\node[scale=0.4] at (15, 0) {|};

\node[scale=0.8] at (20,-2) {40};
\node[scale=0.4] at (20,0) {|};

\node[scale=0.8] at (25,-2) {50};
\node[scale=0.4] at (25,0) {|};

\node[scale=0.8] at (30,-2) {60};
\node[scale=0.4] at (30,0) {|};

\node[scale=0.8] at (35,-2) {70};
\node[scale=0.4] at (35,0) {|};

\node[scale=0.8] at (40,-2) {80};
\node[scale=0.4] at (40,0) {|};

\node[scale=0.8] at (45,-2) {90};
\node[scale=0.4] at (45,0) {|};

\node[scale=0.8] at (50,-2) {100};
\node[scale=0.4] at (50,0) {|};
\node[scale=0.8] at (53,-1) {$n$};
\end{tikzpicture}
\caption{Performance of the model with high density.}\label{fig4}
\end{minipage}
\end{figure}

In Figure \ref{fig1} is pictured the mean time to obtain an optimal solution, for a fixed instance size $n=50$, for each density. In Figures \ref{fig2}, \ref{fig3} and \ref{fig4} are pictured the mean time for following density, low,  medium and  high respectively.
The  above  linear formulation solved almost all the instances in a reasonable amount of time in the mean. We can see that in the case of fixed instances size $n$,   the mean time to obtain an optimal solution increases  in parallel with the density of the compatibility graph, see Figure \ref{fig1}. As expected, Figures \ref{fig2}, \ref{fig3} and \ref{fig4} indicated that for a fixed density of the graph $G$, the  mean time to obtain an optimal solution increases  in parallel with the number of jobs. 

In conclusion, notice that for $n=200$, CPLEX solver solves $60\%$ of the generated instances in $900$ seconds, for $n=100$ it solves $83\%$ and for $n=50$ it solves $100\%$ in few time. All the tests carried out have shown the effectiveness of the model.

\section{Heuristics for the problem $Bm, max|G=(V,E),b=2, s|C_{max}$}

In this section, we propose two heuristics H1 and H2 for the sake of obtaining approximate solutions to the problem $Bm, max|G=(V,E),b=2, s|C_{max}$. The pseudo-code of
each heuristic is preceded by a brief description of the idea behind it.

\subsection{Maximum weighted matching based heuristic (H1)}

This heuristic starts with an initial feasible solution corresponding to the maximum weighted matching in the graph $H_{\alpha}=(G,\alpha)$,  where each edge is valued by the minimum between the processing times of these corresponding jobs. We schedule the batches at the first available machine according to the decreasing order of their processing times.  At each iteration of the heuristic algorithm, we make a random permutation between the jobs of the batches respecting the compatibility constraint,  we swap randomly  the batches of the machines and we  save the best obtained solution.

\begin{algorithm}[h!]
\footnotesize 
\caption{\textbf{H1}}
\KwIn{$G=(V, E)$, $p_i$, $s$} 
\KwResult{ $\sigma_1$}
Apply the Algorithm \ref{alg5} (let $C_{max}$ be the value of the solution given by Algorithm IS)\;
 Let $C=\lbrace M_k, C_{max}^k=C_{max}\rbrace$\;
\For {each $M_k \in C$ }
{\For{each two batches $B_i=(J_i^1, J_i^2)$ and $B_j=(J_j^1, J_j^2)$ in $M_k$}
  {\If {$p_i^1<p_j^2$ and $p_j^1<p_i^2$}
  {\If {$(J_i^{2}, J_j^{2}) \in E$ and there is a machine $ M_l$ where $C_{max}^l+ Min \lbrace p_i^1, p_j^1 \rbrace +s \leq C_{max} $}
  {
  Replace the batches $B_i$ and $B_j$ of $M_k$ by the batch $(J_i^2, J_j^2)$ and a batch that contains one of the jobs $J_i^1$ or $J_j^1$  of the  maximum processing time, and schedule the other job on the machine $M_l$\;
  
  }
  }
  } 
Select a machine $M_r$ such that $ C_{max}^r <C_{max}$ and two batches $B_i$ in $M_k$ and $B_j$ in $M_r$;\\
  \If {$p_i^1<p_j^2$ and $p_j^1<p_i^2$}
  {\If {$(J_i^{2}, J_j^{2}) \in E$ and there is a machine $ M_l$ where $C_{max}^l+ Min \lbrace p_i^1, p_j^1 \rbrace +s \leq C_{max} $}
  {
  Replace the batch $B_i$  by a batch that contains one of the jobs $J_i^1$ or $J_j^1$  of the  maximum processing time\;
  Replace the batch $B_j$  by the batch $(J_i^2,J_j^2)$\;
 Add to the machine $M_l$ a batch that contains the other job\;
  Save the best solution\;
  
  }
  
  }
  
  }
\For {every two machines $M_k$ and $M_l$} 
{
Select randomly a job $J_i^a$ $\in$ $B_i$ and a job $J_j^b$ $\in$ $B_j$\;
\If {$(J_i^{a}, J_j^{b}) \in E$}
  {
Select randomly a number $r\in \lbrace k,l\rbrace$, and add the batch $(J_i^a,J_j^b)$ to the machine $M_r$ and add to the other machine a batch which contains the two remaining jobs if they are adjacent, otherwise schedule each one in a batch\;
Save the best solution\;
  }
}
Form a batch list of a single job\;
Arrange the jobs of the list in the decreasing order of their processing times\;
Select the job at the top of the list and place it at the first available batch. If the job cannot be processed in any one of the existing batches, create a new batch\;
Reschedule all batches on the first available machine according to the decreasing order of their processing times\;
\For{ $counter =1..iter$}
{
 Select randomly  two machines  $M_k$ and $M_l$;\\
 Select randomly a batch $B_i$ $\in$ $M_k$ and a batch $B_j$ $\in$ $M_l$;\\ 
 Swap the two batches of the two machines;\\ 
 Save the best solution;
}
\end{algorithm}


\subsection{Long processing times heuristic (H2)}

This heuristic based on the sequence of  jobs ordered according to the decreasing order of their processing times. We select the job at the top of the list and place it in the first available batch respecting the compatibility constraint. If the job cannot be processed in any one of the existing batches, create a new batch. We schedule the batches on the first available machine and for each iteration, we swap  randomly two batches of the two machines. We save the best solution.

\begin{algorithm}[h!]
\caption{\textbf{H2}}
\KwIn{$G=(V, E)$, $p_i$, $s$} 
\KwResult{ $\sigma_2$}

Arrange the jobs in the decreasing order of their processing times (resulting in
jobs $J_1,...,J_n$);\\
$h:=1$; $B_h:=J_1$;\\
\For {$i=2$ to $n$}
{
\If{(there exists a batch (let $B_j$)such that $\vert B_j\vert =1$ and the job of this batch is
compatible with the job $J_i$)}
{
$B_j:=B_j\cup J_i$;\\
\Else{$h:=h+1;$\\
$B_h:=J_i$;}
}

}

\For{ counter $=1$ to $iter$}
{
 Select randomly  two machines  $M_k$ and $M_l$;\\
 Select randomly a batch $B_i$ $\in$ $M_k$ and a batch $B_j$ $\in$ $M_l$;\\ 
 Swap the two batches of the two machines;\\ 
 Save the best solution;
}
\end{algorithm}

\section{Experimental results}

In this section, we present the results of the computational experiments performed.  We first give details about both software and hardware setup used to run those experiments. Then, we examine the performance of the heuristic approaches H1 and H2.

We coded the above heuristics in C++ language. The tests were executed on a personal computer with an Intel(R) Core(TM) i3-5010U 2.10GHz processor with 4GB of RAM and run under the Windows 7
operating system. We constructed  problem instances with randomly generated data. We considered the number of jobs $n\in\lbrace10, 20, 30, 40, 50, 60, 70, 80, 90, 100, 200, 300, 400\rbrace$, the number of machines $m\in\lbrace 2, 3, 4, 5\rbrace$ and we have used the parameter $d\in\lbrace 12.5, 25, 50, 75, 100\rbrace$ which is the density (in percentage) of the compatibility graph.  The processing times and  the set-up times  are randomly generated in  $\left[ 10, 100\right]$ and $\left\lbrace 2, 3, 4\right\rbrace$ respectively. For each number $n$ of jobs,  we considered 50 of randomly generated instances  according to the uniform distribution.

Table \ref{table1} that summarizes the experimental study shows the efficiency of  H1 and H2. We recorded the minimum, mean and maximum errors for the results 
produced by the heuristics. The performance ratio of a solution $\sigma$ is denoted by $GAP(\sigma)=\left( C_{max}(\sigma)-BS \right)/BS$, where $C_{max}(\sigma)$ and $BS$ are the maximum completion times of the proposed methods and the best solution value obtained by Cplex, respectively. The running times of the algorithms were reported in minimum, mean and maximum  in seconds, and for each heuristic we calculated the number of best solutions ($\#$sol). 

The experimental results presented in Table \ref{table1} and Figures \ref{fig12}, \ref{fig13} and \ref{fig14}, clearly show that  H1 and H2  give near-optimal solutions, and that  H1 is more efficient than H2  in the terms of the  number of the best solutions and gap.  However, it is clear that H1 consumes much more time, compared to H2, to provide an approximate solution. Prevented us from using H2.

\begin{table}[h!]
\footnotesize 
\caption{The performance of the heuristics.}\label{table1}
\begin{center}\renewcommand{\arraystretch}{1}
\begin{tabular}{c c c c c c c c}
\hline
$n$    & Fnct& $\#$Sol H1 &time(s) H1&GAP1 &  $\#$Sol H2 & time(s) H2&GAP2 \\

\hline
10& Min& 50& 11.201& 0&50&2,246&0 \\
   &Mean& 50& 15.667&0&50&2,6886&0 \\
   &Max& 50& 24.539& 0&50&3,229&0\\ \hline

20& Min& 49& 12.745& 0&15&2,761&0 \\
   &Mean& 49& 16.542&0.002&15&3,365&0.010 \\
   &Max& 49& 35.459& 0.020&15&4,352&0.062\\ \hline

30& Min& 46& 15.257& 0&36&3.058&0 \\
   &Mean& 46& 36.108&0.001&36&3.727&0.003 \\
  &Max&46& 44.304& 0.012&36&4.477&0.020\\ \hline

40& Min& 49& 33.306& 0&2&3,463&0 \\
   &Mean&49&41.086&0.003&2&4,385&0.011 \\
  &Max&49& 51.121& 0.023&2&5,164&0.040\\ \hline

50& Min&49&37.190& 0&4&4,181&0 \\
  &Mean&49&49.613&0.002&4&5,259&0.020 \\
  &Max&49&61.979& 0.007&4&6,536&0.045\\ \hline

60& Min& 43&44.866& 0&38&4,259&0 \\
  &Mean&43&58.135&0.001&38&5,662&0.002 \\
  &Max&43&69.888& 0.011&38&7,145&0.014\\ \hline

70& Min& 46&30.576& 0&12&4,992&0 \\
  &Mean& 46&46.303&0.002&12&6,634&0.007 \\
  &Max&46& 75.067& 0.010&12&7,722&0.018\\ \hline

80& Min& 49&34.133& 0&7&5,678&0 \\
  &Mean&49&54.321&0.002&7&7,592&0.013 \\
  &Max&49&80.637& 0.010&7&9,391&0.031\\ \hline

90& Min& 19&38.423& 0&45&5,288&0 \\
  &Mean& 19&49.759&0.002&45&7,813&0.001 \\
  &Max& 19&69.561& 0.011&45&9,469&0.005\\ \hline

100& Min&47& 41.324& 0&8&5,460&0.003 \\
  &Mean&47&89.283&0.003&8&8,030&0.010 \\
  &Max&47&109.777& 0.010&8&10,312&0.019\\ \hline

200& Min&40&130.962& 0.0003&14&9,266&0.0006 \\
  &Mean&40&309.600&0.002&14&13,772&0.004 \\
  &Max&40&534.722& 0.009&14&19,968&0.012\\ \hline

300& Min& 13&554.984& 0.0003&48&16,759&0.0002 \\
  &Mean&13&1305.073&0.002&48&24,172&0.001 \\
  &Max&13&3929.04& 0.005&48&32,869&0.005\\ \hline

400& Min& 42&1542,734& 0&14&21,746&0.0009 \\
  &Mean& 42&4995,794&0.001&14&33,706&0.003 \\
  &Max&42& 12955,851& 0.005&14&47,330&0.004\\ \hline
\end{tabular}
\end{center}
\end{table}

\begin{figure}[h!]
\begin{minipage}{0.325\linewidth}
\begin{tikzpicture}[x=0.7mm,y=1mm,scale=0.9]
\color{black}
\draw[->,>=latex] (0, 0) -- (60, 0) node[below]{};
\draw[->,>=latex] (0, 0) -- (0,38) node[left]{};

\draw plot[smooth,very thick,xscale=0.14,yscale=0.007]coordinates{(10,15.667)(20,16.542) (30,36.108) (40,41.086)(50,49.613)(60,58.135) (70,46.303) (80,54.321) (90,49.759) (100,89.283) (200,309.6) (300,1305.073) (400,4995.794)};

\draw[red] plot[smooth,very thick,xscale=0.14,yscale=0.007]coordinates{(10,2.688)(20,3.365) (30,3.727) (40,4.385)(50,5.259)(60,5.662) (70,6.634) (80,7.592) (90,7.813) (100,8.03) (200,13.772) (300,24.172) (400,33.706)};


\node[scale=0.8] at (0.14*20, -2.5) {20};
\node[scale=0.4] at (0.14*20, 0) {|};

\node[scale=0.4] at (0.14*40,0) {|};


\node[scale=0.4] at (0.14*60,0) {|};


\node[scale=0.4] at (0.14*80,0) {|};


\node[scale=0.8] at (0.14*100,-2.5) {100};
\node[scale=0.4] at (0.14*100,0) {|};

\node[scale=0.8] at (0.14*200,-2.5) {200};
\node[scale=0.4] at (0.14*200,0) {|};

\node[scale=0.8] at (0.14*300,-2.5) {300};
\node[scale=0.4] at (0.14*300,0) {|};

\node[scale=0.8] at (0.14*400,-2.5) {400};
\node[scale=0.4] at (0.14*400,0) {|};
\node[scale=0.8] at (0.14*440,-1) {$n$};

\node[scale=0.8] at (-6,0.007*1000) {1000};
\node[scale=1] at (0,0.007*1000) {-};

\node[scale=0.8] at (-6,0.007*2000) { 2000};
\node[scale=1] at (0,0.007*2000) {-};

\node[scale=0.8] at (-6,0.007*3000) {3000};
\node[scale=1] at (0,0.007*3000) {-};

\node[scale=0.8] at (-6,0.007*4000) {4000};
\node[scale=1] at (0,0.007*4000) {-};

\node[scale=0.8] at (-6,0.007*5000) {5000};
\node[scale=1] at (0,0.007*5000) {-};
\node[scale=0.8] at (0,0.007*5600) {time (s)};

\node[scale=0.8] at (-5,0.007*500) {500};
\node[scale=1] at (0,0.007*500) {-};

\node[scale=1] at (0,0.007*250) {-};

\node[scale=0.8] at (-5,0.007*90) {100};
\node[scale=1] at (0,0.007*100) {-};

\node[] at (405*0.14,0.007*5200) {\textbf{H1}};

\node[red] at (405*0.14,0.007*250) {\textbf{H2}};

\end{tikzpicture}
\caption{H1 vs. H2 for times.}\label{fig12}
\end{minipage}
\begin{minipage}{0.325\linewidth}
\begin{tikzpicture}[x=0.7mm,y=1.3mm,scale=1]
\color{black}
\draw[->,>=latex] (0, 0) -- (60, 0) node[below]{};
\draw[->,>=latex] (0, 0) -- (0,27) node[left]{};

\draw plot[smooth,very thick,xscale=0.14,yscale=0.5]coordinates{(10,50)(20,49) (30,46) (40,49)(50,49)(60,43) (70,46) (80,49) (90,19) (100,47) (200,40) (300,13) (400,42)};

\draw[red] plot[smooth,very thick,xscale=0.14,yscale=0.5]coordinates{(10,50)(20,15) (30,36) (40,3)(50,4)(60,38) (70,12) (80,7) (90,45) (100,8) (200,14) (300,48) (400,14)};


\node[scale=0.8] at (0.14*20, -2) {20};
\node[scale=0.4] at (0.14*20, 0) {|};

\node[scale=0.4] at (0.14*40,0) {|};


\node[scale=0.4] at (0.14*60,0) {|};


\node[scale=0.4] at (0.14*80,0) {|};


\node[scale=0.8] at (0.14*100,-2) {100};
\node[scale=0.4] at (0.14*100,0) {|};

\node[scale=0.8] at (0.14*200,-2) {200};
\node[scale=0.4] at (0.14*200,0) {|};

\node[scale=0.8] at (0.14*300,-2) {300};
\node[scale=0.4] at (0.14*300,0) {|};

\node[scale=0.8] at (0.14*400,-2) {400};
\node[scale=0.4] at (0.14*400,0) {|};

\node[scale=0.75] at (0.14*435,-1) {$n$};

\node[scale=0.8] at (-3.5,0.5*5) {5};
\node[scale=1] at (0,0.5*5) {-};

\node[scale=0.8] at (-3.5,0.5*10) { 10};
\node[scale=1] at (0,0.5*10) {-};

\node[scale=0.8] at (-3.5,0.5*15) {15};
\node[scale=1] at (0,0.5*15) {-};

\node[scale=0.8] at (-3.5,0.5*20) {20};
\node[scale=1] at (0,0.5*20) {-};

\node[scale=0.8] at (-3.5,0.5*25) {25};
\node[scale=1] at (0,0.5*25) {-};

\node[scale=0.8] at (-3.5,0.5*30) {30};
\node[scale=1] at (0,0.5*30) {-};

\node[scale=0.8] at (-3.5,0.5*35) {35};
\node[scale=1] at (0,0.5*35) {-};

\node[scale=0.8] at (-3.5,0.5*40) {40};
\node[scale=1] at (0,0.5*40) {-};

\node[scale=0.8] at (-3.5,0.5*45) {45};
\node[scale=1] at (0,0.5*45) {-};

\node[scale=0.8] at (-3.5,0.5*50) {50};
\node[scale=1] at (0,0.5*50) {-};
\node[scale=0.8] at (0,0.5*56) {$\#Sol$};

\node[] at (405*0.14,0.5*45) {\textbf{H1}};

\node[red] at (405*0.14,0.5*12) {\textbf{H2}};

\end{tikzpicture}
\caption{H1 vs. H2 for $\#$Sol.}\label{fig13}
\end{minipage}
\begin{minipage}{0.325\linewidth}
\begin{tikzpicture}[x=0.7mm,y=1mm,scale=1]
\color{black}
\draw[->,>=latex] (0, 0) -- (60, 0) node[below]{};
\draw[->,>=latex] (0, 0) -- (0,32) node[left]{};

\draw plot[smooth,very thick,xscale=0.14,yscale=150]coordinates{(10,0)(20,0.02) (30,0.01) (40,0.03)(50,0.02)(60,0.01) (70,0.02) (80,0.02) (90,0.02) (100,0.03) (200,0.02) (300,0.02) (400,0.01)};

\draw[red] plot[smooth,very thick,xscale=0.14,yscale=150]coordinates{(10,0)(20,0.1) (30,0.03) (40,0.11)(50,0.2)(60,0.02) (70,0.07) (80,0.13) (90,0.01) (100,0.1) (200,0.04) (300,0.01) (400,0.03)};


\node[scale=0.8] at (0.14*20, -2.5) {20};
\node[scale=0.4] at (0.14*20, 0) {|};

\node[scale=0.4] at (0.14*40,0) {|};


\node[scale=0.4] at (0.14*60,0) {|};


\node[scale=0.4] at (0.14*80,0) {|};


\node[scale=0.8] at (0.14*100,-2.5) {100};
\node[scale=0.4] at (0.14*100,0) {|};

\node[scale=0.8] at (0.14*200,-2.5) {200};
\node[scale=0.4] at (0.14*200,0) {|};

\node[scale=0.8] at (0.14*300,-2.5) {300};
\node[scale=0.4] at (0.14*300,0) {|};

\node[scale=0.8] at (0.14*400,-2.5) {400};
\node[scale=0.4] at (0.14*400,0) {|};
\node[scale=0.8] at (0.14*435,-1) {$n$};

\node[scale=0.8] at (-5.5,150*0.1) {0.01};
\node[scale=1] at (0,150*0.1) {-};

\node[scale=1] at (0,150*0.11) {-};

\node[scale=0.8] at (-6,150*0.07) {0.007};
\node[scale=1] at (0,150*0.07) {-};

\node[scale=1] at (0,150*0.04) {-};

\node[scale=0.8] at (-6,150*0.13) {0.013};
\node[scale=1] at (0,150*0.13) {-};

\node[scale=0.8] at (-5.5,150*0.2) { 0.02};
\node[scale=1] at (0,150*0.2) {-};
\node[scale=0.8] at (0,150*0.22) { GAP};

\node[scale=0.8] at (-6,150*0.03) {0.003};
\node[scale=1] at (0,150*0.03) {-};

\node[scale=1] at (0,150*0.02) {-};

\node[scale=0.8] at (-5.5,150*0.01) {0.001};
\node[scale=1] at (0,150*0.01) {-};

\node[] at (420*0.14,150*0.01) {\textbf{H1}};

\node[red] at (415*0.14,150*0.04) {\textbf{H2}};

\end{tikzpicture}
\caption{H1 vs. H2 for GAP.}\label{fig14}
\end{minipage}
\end{figure}

\section{Conclusion}
We have considered in this paper the problem of scheduling  on batch processing identical machines. We have established some polynomial results for the problem of minimizing the makespan on a single and several machines and have presented Polynomial time algorithms to optimally solve the  problem when the  capacity of batch equal to two. We have presented a linear formulation and we study their performance. We also developed two heuristic approaches to solve the general problem. We conducted numerical experimental tests to study their performance by comparing them to the best solution value obtained by Cplex.
For future research, we can studied the complexity of the NP-hard problems for three values of processing times in the case of job multiplicity. One could also explore exact methods to solve the general case of the problem and compare their performance with the linear formulation proposed above.

%
%
%

\nocite{*}
\bibliographystyle{fundam}
\bibliography{biblio}
\end{document}